\documentclass[oneside,english]{amsart}
\usepackage{lmodern}
\usepackage[T1]{fontenc}
\usepackage[latin9]{inputenc}
\usepackage{geometry}
\usepackage{float}
\usepackage{mathtools}
\usepackage{enumitem}
\usepackage{amsthm}
\usepackage{amssymb}
\usepackage{mathdots}
\usepackage{graphicx}

\makeatletter

\providecommand{\tabularnewline}{\\}
\floatstyle{ruled}
\newfloat{algorithm}{tbp}{loa}
\providecommand{\algorithmname}{Algorithm}
\floatname{algorithm}{\protect\algorithmname}

\numberwithin{equation}{section}
\numberwithin{figure}{section}
\theoremstyle{plain}
\newtheorem{thm}{\protect\theoremname}
\theoremstyle{definition}
\newtheorem{defn}[thm]{\protect\definitionname}
\newlist{casenv}{enumerate}{4}
\setlist[casenv]{leftmargin=*,align=left,widest={iiii}}
\setlist[casenv,1]{label={{\itshape\ \casename} \arabic*.},ref=\arabic*}
\setlist[casenv,2]{label={{\itshape\ \casename} \roman*.},ref=\roman*}
\setlist[casenv,3]{label={{\itshape\ \casename\ \alph*.}},ref=\alph*}
\setlist[casenv,4]{label={{\itshape\ \casename} \arabic*.},ref=\arabic*}
\theoremstyle{remark}
\newtheorem{rem}[thm]{\protect\remarkname}
\theoremstyle{plain}
\newtheorem{lem}[thm]{\protect\lemmaname}

\usepackage{amsmath}
\usepackage{colortbl}
\usepackage{arydshln}
\usepackage{color}
\usepackage{hyperref}

\@ifundefined{showcaptionsetup}{}{%
 \PassOptionsToPackage{caption=false}{subfig}}
\usepackage{subfig}
\makeatother

\usepackage{babel}
\usepackage{listings}
\providecommand{\casename}{Case}
\providecommand{\definitionname}{Definition}
\providecommand{\lemmaname}{Lemma}
\providecommand{\remarkname}{Remark}
\providecommand{\theoremname}{Theorem}

\begin{document}
\title{A Blind Permutation Similarity Algorithm}
\author{Eric Barszcz\\
NASA Ames Research Center\\
Moffett Field, CA 94035\\
}
\begin{abstract}
This paper introduces a polynomial blind algorithm that determines
when two square matrices, $A$ and $B$, are permutation similar.
The shifted and translated matrices $\left(A+\beta I+\gamma J\right)$
and $\left(B+\beta I+\gamma J\right)$ are used to color the vertices
of two square, edge weighted, rook's graphs. Then the orbits are found
by repeated symbolic squaring of the vertex colored and edge weighted
adjacency matrices. Multisets of the diagonal symbols from non-permutation
similar matrices are distinct within a few iterations, typically four
or less.
\end{abstract}

\maketitle

\section{Introduction}

This paper introduces a polynomial blind algorithm for permutation
similarity. Specht's Theorem gives an infinite set of necessary and
sufficient conditions for two square matrices to be unitarily similar
\cite[pg. 97]{hornAndJohnson2013}. Specht's Theorem compares the
traces of certain matrix products. If the traces ever fail to match,
then the matrices are not unitarily similar. Subsequent work reduced
the number of traces required to a finite number \cite[pg. 98]{hornAndJohnson2013}.
The permutation similarity algorithm is somewhat analogous, an infinite
set of matrix products are checked to see if the multisets (not traces)
of the diagonal elements match. If they ever fail to match, the two
matrices are not permutation similar.

Permutation similarity plays a large role in the graph isomorphism
problem where two graphs are isomorphic iff their adjacency matrices
are permutation similar. The best theoretical results for graph isomorphism
are by Babai, showing graph isomorphism is quasipolynomial \cite{babai2015}.
The typical graph isomorphism algorithm uses equitable partitions
and vertex individualization with judicious pruning of the search
tree to generate canonical orderings that are compared to determine
if two graphs are isomorphic \cite{mckay2014}. However, Neuen and
Schweitzer suggest that all vertex individualization based refinement
algorithms have an exponential lower bound \cite{neuen2017a}. The
permutation similarity algorithm does not perform vertex individualization
nor does it try to construct a canonical ordering.

The overall process is as follows. Square (real or complex) matrices
$A$ and $B$ are converted into positive integer matrices whose diagonal
entries that are distinct from the off-diagonal entries, $\left(A_{\Sigma}+\beta I+\gamma J\right)$
and $\left(B_{\Sigma}+\beta I+\gamma J\right)$. The shifted and translated
matrices are used to color the vertices of edge weighted rook's graphs.
The edge weights are `1' for a column edge and `2' for a row edge.
The resultant graphs are called permutation constraint graphs (PCGs),
see Figure \ref{fig:rookVsPCG}. The purpose of a PCG is to add symmetric
permutation constraints to the original matrix. Then the vertex colored
and edged weighted adjacency matrices of the PCGs are constructed.
Such an adjacency matrix is called a permutation constraint matrix
(PCM). It will be shown that $A$ and $B$ are permutation similar
iff their associated PCMs are permutation similar. Next, the PCMs
are repeatedly squared using symbolic matrix multiplication. Symbolic
squaring generates a canonical string for each inner product. Symbols
are substituted for the strings in a consistent fashion. After the
symbol substitution, the multisets of diagonal symbols are compared.
If they every differ, the matrices are not permutation similar. Symbolic
squaring monotonically refines the partitions until a stable partition
is reached. 

Experimentally, a stable partition is reached in six iterations or
less and non permutation similar matrices are separated in four iterations
or less. The canonical inner product strings generated by symbolic
squaring separate the orbits imposed by the automorphism group acting
on the set of $(i,j)$ locations.

The necessity of the blind algorithm is proven in Section \ref{sec:Blind-P-similarity-Algorithm}.
Sufficiency is argued in Section \ref{sec:The-Crux}. Section \ref{sec:FindingP}
provides a polynomial algorithm that uses the blind algorithm to find
a permutation between permutation similar matrices. Section \ref{sec:Discussion}
provides some additional material that readers may find interesting.
Section \ref{sec:Summary} is a summary.

\section{\label{sec:Blind-P-similarity-Algorithm}Blind P-similarity Algorithm}

Algorithm \ref{alg:Blind-p-similarity-algorithm} contains pseudocode
for the blind permutation similarity algorithm. This section is organized
so the background for each function is discussed in order. The discussion
will make it clear that the algorithm tests necessary conditions for
two matrices to be permutation similar. A conservative complexity
bound for Algorithm \ref{alg:Blind-p-similarity-algorithm} is presented
in Section \ref{subsec:Algorithm-Complexity}. Sufficiency arguments
are given in Section \ref{sec:The-Crux}.

\begin{algorithm}
\begin{lstlisting}
01 function psim = BPSAY(M1,M2)
02 % Blind Permutation Similarity Algoritm, Yes?
03 % Inputs:
04 %   M1 - mxm square matrix (real or complex)
05 %   M2 - mxm square matrix (real or complex)
06 % Outputs:
07 %   psim - boolean (TRUE if M1 & M2 are p-similar)
08 
09 psim = FALSE;	% initialize psim to FALSE
10 
11 [A,B] = SymbolSubstitution(M1, M2);
12
13 S = PCM(ShiftAndTranslate(A, beta=m^2, gamma=2));
14 T = PCM(ShiftAndTranslate(B, beta=m^2, gamma=2));
15 
16 repeat
17 	[Ssqr,Tsqr] = SymbolSubstitution(SymSqr(S), SymSqr(T));
18 	if (~CompareDiagMultisets(Ssqr,Tsqr)), break; end if	% not p-sim
19
20 	converged = ComparePatterns(S,Ssqr) & ComparePatterns(T,Tsqr);
21 	if (converged), psim = TRUE; end if			% p-sim
22 
23 	S = Ssqr;  T = Tsqr;
24 until (converged)
25 
26 return(psim);
27 end
\end{lstlisting}

\caption{\label{alg:Blind-p-similarity-algorithm}Blind p-similarity algorithm
pseudocode}
\end{algorithm}

\subsection{Permutation Similarity}

Real or complex square $m\times m$ matrices $A$ and $B$ are permutation
similar iff there exists a permutation matrix $P$ such that $PAP^{T}=B$
\cite[pg. 58]{hornAndJohnson2013}. The notation $A\underset{P}{\sim}B$
is used to indicate that $A$ and $B$ are permutation similar. Permutation
similarity will be abbreviated as p-similarity in the text.

\subsection{\label{subsec:Symmetric-Permutation}Symmetric Permutation and Mixes}

The process of applying a permutation matrix $P$ from the left and
right as $PAP^{T}$ is called a symmetric permutation \cite[pg. 147]{golub}.
Symmetric permutations have four main characteristics:
\begin{enumerate}
\item Elements of a matrix are moved around by symmetric permutation but
not changed, so the multiset of elements in a matrix before and after
a symmetric permutation are identical. Let the multiset of all elements
in a matrix be called the \emph{mix}, i.e.,\emph{ }$mix\left(A\right)=\left\{ \left\{ A_{i,j}\right\} \right\} $.
Then $mix\left(A\right)=mix\left(PAP^{T}\right)$.
\item Elements of a row are moved together, so the multiset of elements
from a row before and after being moved by a symmetric permutation
are identical. Let the multiset of row multisets be called the \emph{row
mix}, i.e.,\emph{ }$rowMix\left(A\right)=\left\{ \left\{ \begin{array}{ccc}
\left\{ \left\{ A_{1,:}\right\} \right\} , & \cdots, & \left\{ \left\{ A_{m,:}\right\} \right\} \end{array}\right\} \right\} $. Then $rowMix\left(A\right)=rowMix\left(PAP^{T}\right)$.
\item Elements of a column are moved together, so the multiset of elements
from a column before and after being moved by a symmetric permutation
are identical. Let the multiset of column multisets be called the
\emph{column mix}, i.e.,\emph{ $colMix\left(A\right)=\left\{ \left\{ \begin{array}{ccc}
\left\{ \left\{ A_{:,1}\right\} \right\} , & \cdots, & \left\{ \left\{ A_{:,m}\right\} \right\} \end{array}\right\} \right\} $}. Then $colMix\left(A\right)=colMix\left(PAP^{T}\right)$.
\item Elements on the diagonal remain on the diagonal, so the multiset of
elements on the diagonal is the same before and after a symmetric
permutation. Let the multiset of diagonal elements be called the \emph{diagonal
mix}, i.e., $diagMix\left(A\right)=\left\{ \left\{ A_{j,j}\right\} \right\} $.
Then $diagMix\left(A\right)=diagMix\left(PAP^{T}\right)$.
\end{enumerate}

\subsection{\label{subsec:Triplet-Notation}Triplet Notation for a Matrix}

Because symmetric permutations move values around without changing
them, the values can be replaced with a different set of symbols without
changing the action. It is useful to view a $m\times m$ matrix as
a triplet consisting of i) $\Pi$: a partition of the $(i,j)$ locations,
ii) $\Sigma$: a set of distinct symbols, and iii) $g$: a bijective
mapping between symbols and cells of the partition. Each cell of the
partition is associated with a single symbol/value. 

For example, let $M$ be a $m\times m$ matrix. Define $L$ as the
set of all $(i,j)$ locations, 
\[
L=\left\{ \left(i,j\right)\,|\,1\leq i\leq m,\,1\leq j\leq m\right\} .
\]
 Then the triplet notation for $M$ is 
\[
M=\left(\Pi,\Sigma,g\right)
\]
 where $\Pi$ is a partition of $L$ where each cell contains all
of the locations associated with the same symbol, $\Sigma$ is the
set of distinct symbols in $M$, and $g$ is a bijective mapping of
symbols in $\Sigma$ to cells in $\Pi$.

Often $\Pi$ and $\Sigma$ are assumed to be ordered. $\Pi$ is ordered
by first applying lexicographic ordering within each cell to order
the locations in the cell. The first location in each ordered cell
is designated the \emph{representative location} for that cell. Next,
the representative locations are collected and lexicographically ordered
to order the cells themselves. The set of symbols $\Sigma$ is ordered
such that the first symbol is associated with the first cell of $\Pi$
and the second symbol with the second cell, and so on. So the mapping
$g$ from $\Sigma$ to $\Pi$ is simply the identity mapping. If the
notation $M=\left(\Pi,\Sigma\right)$ appears, it means that $\Pi$
and $\Sigma$ are ordered and $g$ is the identity mapping.

To compare two partitions, $\Pi1$ and $\Pi2$, the representative
location for each cell is assigned to all locations within the cell,
essentially using the representative locations as the symbol set.
Then the \emph{partition difference}, $\Pi1-\Pi2$, is defined as
\[
\left(\Pi1-\Pi2\right)_{i,j}=\begin{cases}
0 & \mbox{if the representative locations at }(i,j)\mbox{ match}\\
1 & \mbox{if the representative locations at }(i,j)\mbox{ do not match.}
\end{cases}
\]

The partition $\Pi$ will be referred to as a \emph{pattern} to evoke
a two dimensional image. $\Pi\left(A\right)$ is a refinement of $\Pi\left(B\right)$,
$\Pi\left(A\right)\preceq\Pi\left(B\right)$, if each cell of $\Pi\left(A\right)$
is a subset of a cell of $\Pi\left(B\right)$ \cite{mckay2014}.

\subsection{\label{subsec:Symbol-Substitution}Symbol Substitution: Lines 11
\& 17}

A consequence of symmetric permutations moving values/symbols around
without changing them is that the symbols can be replaced without
changing the p-similarity relation.
\begin{thm}
\label{thm:ConsistentSymbolSubstitution}Given square matrices $A$
and $B$. $A\underset{P}{\sim}B\iff A_{\overline{\Sigma}}\underset{P}{\sim}B_{\overline{\Sigma}}$
where $A=\left(\Pi_{A},\Sigma,g_{A}\right)$, $B=\left(\Pi_{B},\Sigma,g_{B}\right)$,
$A_{\overline{\Sigma}}=\left(\Pi_{A},\overline{\Sigma},g_{A}\right)$,
$B_{\bar{\Sigma}}=\left(\Pi_{B},\overline{\Sigma},g_{B}\right)$ and
$\left|\Sigma\right|=\left|\overline{\Sigma}\right|$.
\end{thm}

\begin{proof}
Assume $A\underset{P}{\sim}B$ so there exists a permutation matrix
$P$ such that $PAP^{T}=B$. Further assume that the symmetric permutation
moves the symbol at location $(i,j)$ of $A$ to location $(r,s)$
in $B$. So $A_{i,j}=B_{r,s}=\sigma_{k}$ where $\sigma_{k}\in\Sigma$
is the $k^{{\rm th}}$ symbol of $\Sigma$. Everywhere $\sigma_{k}$
appears in $A$ and $B$ will be replaced by $\overline{\sigma}_{k}$,
the $k^{{\rm th}}$ symbol of $\overline{\Sigma}$ in $A_{\overline{\Sigma}}$
and $B_{\overline{\Sigma}}$. Similarly for each of the other symbols
in $\Sigma$. Therefore, $PA_{\overline{\Sigma}}P^{T}=B_{\overline{\Sigma}}$
so $A_{\overline{\Sigma}}\underset{P}{\sim}B_{\overline{\Sigma}}$.

For the other direction, reverse the roles of $\Sigma$ and $\overline{\Sigma}$.
\end{proof}
A consequence of Theorem \ref{thm:ConsistentSymbolSubstitution} is
that the p-similarity problem for real or complex matrices reduces
to solving p-similarity for positive integer matrices. 

The symbol substitution function, $SymbolSubstitution\left(\,\right)$,
appears on lines 11 and 17 in Algorithm \ref{alg:Blind-p-similarity-algorithm}.
The function takes two square arrays of values (real or complex numbers
in line 11, and strings in line 17) and outputs two square, positive
integer matrices. The substitutions are performed so that distinct
input values are assigned distinct output symbols and identical input
values are assigned identical output symbols.

Given $m\times m$ arrays as input, a complexity bound on symbol substitution
is based on assumption that it is performed similar to the Sieve of
Eratosthenes. Arrays are searched in column major order. When a previously
unseen input value is encountered, a new output symbol is assigned
to that input value and $O\left(m^{2}\right)$ comparisons are performed
on each input array to find all locations containing that input value.
Once all matching input values have been replaced with the output
symbol, the next unseen input value is located and the process repeats.
Since there are at most $2m^{2}$ distinct input values, there are
$O\left(m^{4}\right)$ comparisons. So a symbol substitution using
this approach requires $O\left(m^{4}\right)$ comparison operations
times the cost to do a comparison operation. 

The sufficiency argument in Section \ref{sec:The-Crux} uses a second
symbol substitution function $SymSub\left(\,\right)$ which takes
a single square array of values and output a positive integer matrix.
For reproducibility, symbol assignments are performed in a permutation
independent fashion. Assume there a $n$ distinct symbols where there
are $n_{1}$ distinct off-diagonal symbols and $n_{2}$ distinct diagonal
symbols, with $n=n_{1}+n_{2}$. The set of $n_{1}$ distinct off-diagonal
strings are lexicographically ordered and then 1 is assigned to the
first string in the ordered list and so on until $n_{1}$ is assigned
to the last string in the ordered list. Similarly, the list of $n_{2}$
distinct diagonal strings are ordered prior to assigning $\left(n_{1}+1\right)$
to $n$ to them. This guarantees that any symmetric permutation applied
to the array of canonical inner product strings will assign the same
integer to the same string.

\subsection{Shifting and Translating: Lines 13 \& 14}

The next theorem shows up in spectral graph literature with integer
coefficients for $\beta$ and $\gamma$. It is presented here in the
more general form where the coefficients are real numbers.
\begin{thm}
\label{thm:ShiftAndTranslate}Given square matrices $A$ and $B$.
$A\underset{P}{\sim}B\iff\left(\alpha A+\beta I+\gamma J\right)\underset{P}{\sim}\left(\alpha B+\beta I+\gamma J\right)$
where $I$ is the identity matrix, $J$ is the matrix of all ones
and $\alpha$, $\beta$, $\gamma\in\mathbb{R}$ with $\alpha\neq0$.
\end{thm}

\begin{proof}
By noting that $PIP^{T}=I$ and $PJP^{T}=J$ for all permutation matrices
$P$ so 
\begin{eqnarray*}
PAP^{T} & = & B\\
\alpha PAP^{T}+\beta I+\gamma J & = & \alpha B+\beta I+\gamma J\\
P\left(\alpha A+\beta I+\gamma J\right)P^{T} & = & \alpha B+\beta I+\gamma J.
\end{eqnarray*}
\end{proof}
There is overlap in the capability to replace symbols between Theorem
\ref{thm:ConsistentSymbolSubstitution} and Theorem \ref{thm:ShiftAndTranslate}.
However, if a symbol appears as a diagonal entry and as an off-diagonal
entry, Theorem \ref{thm:ShiftAndTranslate} can make the diagonal
entry distinct from the off-diagonal entry by shifting the spectrum,
changing the number of occurrences of that symbol. Theorem \ref{thm:ConsistentSymbolSubstitution}
cannot change the number occurrences of a symbol, it can only replace
a symbol with a new symbol without changing the the number of occurrences. 

The shift and translate function, $ShiftAndTranslate(\,)$ appears
in lines 13 and 14 of Algorithm \ref{alg:Blind-p-similarity-algorithm}.
The function takes a square, positive integer matrix $M$ and positive
integer values for $\beta$ and $\gamma$ as input and generates the
square, positive integer matrix $\left(M+\beta I+\gamma J\right)$
as output.

\subsection{Color Matrix: Lines 11, 13 \& 14}

Given $m\times m$ matrices $A$ and $B$, define their \emph{color
matrices} as the result of performing a consistent symbol substitution
followed by shifting and translating so that all symbols are positive
integers greater than or equal to three and the symbols on the diagonal
are distinct from all off-diagonal symbols. In particular, assume
$A=\left(\Pi_{A},\Sigma,g_{A}\right)$ and $B=\left(\Pi_{B},\Sigma,g_{B}\right)$
where $\left|\Sigma\right|=k$ is less than or equal to $m^{2}$.
Perform a consistent symbol substitution using $\overline{\Sigma}=\left\{ 1,2,\ldots,k\right\} $
to get $A_{\overline{\Sigma}}$ and $B_{\overline{\Sigma}}$ followed
by shifting and translating using $\beta=m^{2}$ and $\gamma=2$ to
get color matrices $A_{C}=\left(A_{\overline{\Sigma}}+m^{2}I+2J\right)$
and $B_{C}=\left(B_{\overline{\Sigma}}+m^{2}I+2J\right)$.

The smallest integer in $\overline{\Sigma}$ is one and the largest
is $k$, so shifting the spectrum of $A_{\overline{\Sigma}}$ and
$B_{\overline{\Sigma}}$ by $m^{2}I$ guarantees that the diagonal
symbols are distinct from off-diagonal symbols. Then adding $2J$
guarantees the smallest symbol is greater than or equal to three.
\begin{thm}
\label{thm:ColorMatrix}Given square matrices $A$ and $B$ and their
associated color matrices $A_{C}$ and $B_{C}$. $A\underset{P}{\sim}B\iff A_{C}\underset{P}{\sim}B_{C}$
.
\end{thm}

\begin{proof}
Result follows from applying Theorem \ref{thm:ConsistentSymbolSubstitution}
to the symbol substitution followed by Theorem \ref{thm:ShiftAndTranslate}
for the shifting and translating.
\end{proof}
Color matrices are formed in Algorithm \ref{alg:Blind-p-similarity-algorithm}
as a two step process. Given square input matrices $A$ and $B$,
the symbol substitutions are performed on line 11, followed by the
shift and translations on lines 13 and 14 for $A_{C}$ and $B_{C}$
respectively.

\subsection{Permutation Constraint Graph}

Given a $m\times m$ matrix $M$. The \emph{permutation constraint
graph} (PCG) associated with $M$ is the $m\times m$ rook's graph
with distinct column and row edge weights and vertex colors matching
$M_{C}$, $M$'s color matrix. It is called a permutation constraint
graph because the weighted edges of the rook's graph along with the
vertex coloring add symmetric permutation constraints to $M$. The
rationale behind PCGs is discussed in Section \ref{subsec:Origin-of-PCGs}.

It is well known that rows and columns of rook's graphs can be permuted
independently and that the rows and columns of a square rook's graph
can be exchanged, so the total number of automorphisms is $2\left(m!\right)^{2}$
\cite{rookGraph}. The automorphism group for a PCG is a subgroup
of the rook's graph automorphism group since all of the edges are
rook's graph edges. PCGs break the symmetry that allows rows to be
exchanged with columns by adding edge weights. A PCG has a column
edge weight of `1' and a row edge weight of `2' as shown in Figure
\ref{fig:PCGsubB}. 
\begin{figure}
\noindent \centering{}\hfill{}\subfloat[$3\times3$ rook's graph]{\begin{centering}
\includegraphics[width=2in]{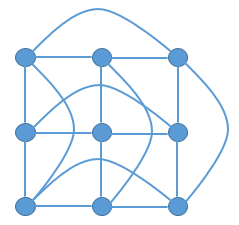}
\par\end{centering}
}\hfill{}\subfloat[\label{fig:PCGsubB}$3\times3$ permutation constraint graph (PCG)]{\begin{centering}
\includegraphics[width=2in]{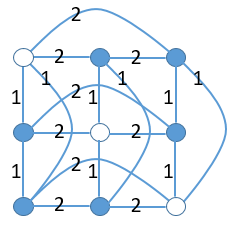}
\par\end{centering}
}\hfill{}\caption{\label{fig:rookVsPCG}$3\times3$ rook's graph versus a $3\times3$
permutation constraint graph}
\end{figure}

To see that a PCG implements all symmetric permutation constraints,
note that rook's graph automorphisms keep rows together and columns
together, similar to symmetric permutations. Symmetric permutation
also requires diagonal elements to remain on the diagonal. PCGs accomplish
this by having the diagonal colors be distinct from off-diagonal colors
as shown in Figure \ref{fig:PCGsubB}. A PCG automorphism must permute
a diagonal entry to another diagonal entry as required by symmetric
permutation. This forces the row permutation and the column permutation
to match, leading to the following theorem.
\begin{thm}
\label{thm:PCGvsPSim}Given square matrices $A$ and $B$. Their respective
PCGs $\Gamma_{A}$ and $\Gamma_{B}$ are isomorphic iff $A\underset{P}{\sim}B$.
\end{thm}

\begin{proof}
Theorem \ref{thm:ColorMatrix} establishes that $A\underset{P}{\sim}B$
iff $A_{C}\underset{P}{\sim}B_{C}$ so if $A\underset{P}{\sim}B$
then $\Gamma_{A}$ and $\Gamma_{B}$ are isomorphic since their vertex
colors are $A_{C}$ and $B_{C}$ are p-similar. In the other direction,
assume PCGs $\Gamma_{A}$ and $\Gamma_{B}$ are isomorphic. The first
step is to show that the re-labeling of vertices takes the form of
a symmetric permutation. $\Gamma_{A}$ and $\Gamma_{B}$ have identical
off-diagonal structure. So any re-labeling must act as an automorphism
for that structure, implying the set of possible re-labelings is a
subset of the rook's graph automorphisms that do not exchange the
rows and columns. Further, since the diagonal colors of $\Gamma_{A}$
and $\Gamma_{B}$ are distinct from off-diagonal colors, the set of
possible re-labelings is further restricted to only include re-labeling
that apply the same permutation to rows and columns. Therefore, any
re-labeling between $\Gamma_{A}$ and $\Gamma_{B}$ is applied as
a symmetric permutation. Otherwise the off-diagonal structure is not
maintained or diagonal vertices of $\Gamma_{A}$ are not mapped to
diagonal vertices of $\Gamma_{B}$. By hypothesis, $\Gamma_{A}$ and
$\Gamma_{B}$ are isomorphic, so there exists a symmetric permutation
mapping vertices of , $\Gamma_{A}$ to vertices of $\Gamma_{B}$ implying
the color matrices $A_{C}$ and $B_{C}$ are p-similar. Then by Theorem
\ref{thm:ColorMatrix} $A$ and $B$ are p-similar.
\end{proof}
Permutation constraint graphs are not explicitly formed in Algorithm
\ref{alg:Blind-p-similarity-algorithm}. However, they provide the
link to the permutation constraint matrices described in the next
section.

\subsection{Permutation Constraint Matrix: Lines 13 \& 14}

The vertex colored adjacency matrix associated with a PCG is called
a \emph{permutation constraint matrix} (PCM). Column major ordering
is used to bijectively map the vertices of the PCG to the diagonal
of the PCM. Let $M_{C}$ be the $m\times m$ color matrix of the PCG.
Then the diagonal of the PCM is given by 
\begin{equation}
D=diag\left(reshape\left(M_{C},m^{2},1\right)\right)\label{eq:D}
\end{equation}
 where $reshape\left(M_{C},m^{2},1\right)$ uses column major ordering
to reshape the $m\times m$ color matrix into a $m^{2}\times1$ vector
and $diag\left(\,\right)$ converts the $m^{2}\times1$ vector into
a $m^{2}\times m^{2}$ diagonal matrix. Applying column major ordering
to the vertices of the PCG creates a regular off-diagonal structure
in the PCM. Let $R$ be the $m^{2}\times m^{2}$ matrix representing
the off-diagonal structure of a PCM. $R$ is given by 
\begin{equation}
R=I_{m}\otimes\left(\mathbf{1}*\left(J_{m\times m}-I_{m}\right)\right)+\left(J_{m\times m}-I_{m}\right)\otimes\left(\mathbf{2}*I_{m}\right)\label{eq:EdgeWeightedR}
\end{equation}
 where $I\otimes\left(\mathbf{1}*\left(J-I\right)\right)$ are the
column edges, $\left(J-I\right)\otimes\left(\mathbf{2}*I\right)$
are the row edges, and $\otimes$ is the Kronecker product. So the
PCM of a $m\times m$ color matrix $M_{C}$ is written as 
\begin{equation}
PCM\left(M_{C}\right)=D+R.\label{eq:PCM}
\end{equation}
 Note that all $m^{2}\times m^{2}$ PCMs have identical off-diagonal
structure. The only difference between two $m^{2}\times m^{2}$ PCMs
is the diagonal.

For example, let $M=J_{3\times3}$, the $3\times3$ matrix of all
ones. The associated color matrix is $M_{C}=\left(J+9I+2J\right)$
and the PCG looks exactly like Figure \ref{fig:PCGsubB}, where the
white diagonal vertices are associated with `12' and the off-diagonal
vertices with `3'. The PCM generated using column major ordering is
\begin{equation}
PCM\left(J+9I+2J\right)=\left[\begin{array}{ccc|ccc|ccc}
12 & 1 & 1 & 2 &  &  & 2\\
1 & 3 & 1 &  & 2 &  &  & 2\\
1 & 1 & 3 &  &  & 2 &  &  & 2\\
\hline 2 &  &  & 3 & 1 & 1 & 2\\
 & 2 &  & 1 & 12 & 1 &  & 2\\
 &  & 2 & 1 & 1 & 3 &  &  & 2\\
\hline 2 &  &  & 2 &  &  & 3 & 1 & 1\\
 & 2 &  &  & 2 &  & 1 & 3 & 1\\
 &  & 2 &  &  & 2 & 1 & 1 & 12
\end{array}\right]\label{eq:Blah}
\end{equation}
where column edges have a weight of `1', row edges a weight of `2',
and the blank areas are filled with zeros. The horizontal and vertical
lines in (\ref{eq:Blah}) are there to emphasize the block structure.
\begin{thm}
\label{thm:ABvsST}Given $m\times m$ matrices $A$ and $B$ and their
associated $m^{2}\times m^{2}$ PCMs, $S=PCM(A_{C})$ and $T=PCM(B_{C})$,
then $A\underset{P}{\sim}B\iff S\underset{P}{\sim}T$.
\end{thm}

\begin{proof}
Using Theorem \ref{thm:PCGvsPSim} and the fact that column major
ordering is a bijective mapping from the $m\times m$ arrays of vertices
of the PCGs to the diagonals of the $m^{2}\times m^{2}$ PCMs.
\end{proof}
If PCMs $S$ and $T$ are p-similar, then the permutation symmetrically
permuting $S$ to $T$ has the form $P\otimes P$. To see this, note
that the adjacency matrix of an edge weighted rook's graph with a
column edge weight of `1' and a row edge weight of `2' is given by
$R$ in (\ref{eq:EdgeWeightedR}). Since rows and columns of $R$
can be permuted independently, but rows cannot be exchanged with columns,
the automorphism group of $R$ is 
\[
Aut\left(R\right)=\left\{ P_{c}\otimes P_{r}\right\} 
\]
 where $P_{c}$ and $P_{r}$ are $m\times m$ permutation matrices
applied to the columns and rows respectively. Since the off diagonal
structure of every $m^{2}\times m^{2}$ PCM is identical to $R$,
any permutation between $S$ and $T$ must have the form $P_{c}\otimes P_{r}$.
However, the vertex coloring on the diagonal of the PCG, when mapped
to the diagonal of the PCM, restrict the possible permutations to
those of the form $P\otimes P$.

Construction of the PCMs from the color matrices occurs in lines 13
and 14 of Algorithm \ref{alg:Blind-p-similarity-algorithm}. The function
$PCM\left(\,\right)$ takes a $m\times m$ color matrix $M_{C}$ as
input and returns the associated $m^{2}\times m^{2}$ PCM defined
by equations (\ref{eq:D}), (\ref{eq:EdgeWeightedR}), and (\ref{eq:PCM})
as output.

To recap, given (real or complex) square matrices $A$ and $B$, a
symbol substitution is performed to convert all values/symbols to
be positive integers. Next the color matrices $A_{C}$ and $B_{C}$
are constructed by shifting and translating the positive integer matrices
so the diagonal symbols differ from the off-diagonal symbols and the
smallest value is greater than or equal to three. The color matrices
are used to color the vertices of edge weighted rook's graphs where
column edges have weight one and row edges weight two creating PCGs.
Then PCMs $S$ and $T$ are constructed by using column major ordering
to map the PCG vertices to the diagonals. Theorem \ref{thm:ABvsST}
shows that $S$ and $T$ are p-similar iff $A$ and $B$ are p-similar.

\subsection{Symbolic Matrix Multiplication}

\subsubsection{Eigenspace Projector Patterns}

PCMs are real symmetric matrices. Therefore they have a unique spectral
decomposition in terms of the distinct eigenvalues and their associated
spectral projectors 
\begin{equation}
A=\sum\lambda_{i}E_{i}\mbox{ and }I=\sum E_{i}\label{eq:SpectralDecomp}
\end{equation}
 where $\lambda_{i}$ is an eigenvalue and $E_{i}$ is its associated
spectral projector \cite[pg. 9]{parlett98}. Assume there are $k$
distinct eigenvalues. Rearrange (\ref{eq:SpectralDecomp}) as
\begin{equation}
A_{1\times m\times m}=\left[\lambda_{1},\cdots,\lambda_{k}\right]\times\left[\begin{array}{c}
E_{1}\\
\vdots\\
E_{k}
\end{array}\right]_{k\times m\times m}\label{eq:A=00003DlambdaE}
\end{equation}
 where the spectral projectors are stacked like a pages in a book
lying on a desk and $A$ is a sheet of paper lying on the desk. Any
symmetric permutation applied to $A$ on the lhs is also applied to
each of the spectral projectors on the rhs. Therefore we can construct
a string from each $(:,i,j)$ column by concatenating the values.
Then substituting distinct symbols for distinct strings yields a pattern.
We call this pattern the \emph{eigenspace projector pattern}.

To find the eigenspace projector pattern you need to compute the eigenvalues
and eigenvectors. This is an iterative process and may be susceptible
to floating point arithmetic errors. Alternatively, the eigenspace
projector pattern can be generated from a stack of powers of $A$
by taking advantage of Wilkinson's observation \cite[pg. 13]{wilkinson65}
that the transpose of the Vandermonde matrix 
\[
V^{T}=\left[\begin{array}{ccc}
1 & \cdots & 1\\
\lambda_{1} & \cdots & \lambda_{k}\\
\vdots & \cdots & \vdots\\
\lambda_{1}^{k-1} & \cdots & \lambda_{k}^{k-1}
\end{array}\right],
\]
 is non-singular. Powers of $A$ can be stacked on the lhs as
\begin{equation}
\left[\begin{array}{c}
I\\
A\\
\vdots\\
A^{k-1}
\end{array}\right]_{k\times m\times m}=\left[\begin{array}{ccc}
1 & \cdots & 1\\
\lambda_{1} & \cdots & \lambda_{k}\\
\vdots & \cdots & \vdots\\
\lambda_{1}^{k-1} & \cdots & \lambda_{k}^{k-1}
\end{array}\right]\times\left[\begin{array}{c}
E_{1}\\
\vdots\\
E_{k}
\end{array}\right]_{k\times m\times m}.\label{eq:A=00003DVxE-1}
\end{equation}
Examining (\ref{eq:A=00003DVxE-1}), one sees that strings constructed
on the lhs of (\ref{eq:A=00003DVxE-1}) yield a pattern identical
to the eigenspace projector pattern since $V^{T}$ is non-singular.
PCMs are symmetric integer matrices so integer arithmetic can be used
to construct the eigenspace projector pattern. However, overflow errors
may occur.

The eigenspace projector pattern is a refinement of the individual
spectral projectors. So given a real symmetric matrix, the eigenspace
projector pattern is the most refined pattern one would expect for
that matrix. In graph isomorphism terms, the eigenspace projector
pattern is like the coarsest equitable partition \cite{mckay80} achievable
using naive color refinement. For two real symmetric matrices to be
p-similar, their eigenspace projector patterns must be p-similar.

\subsubsection{Computing Eigenspace Projector Patterns for SPD Matrices}

For symmetric positive definite (SPD) matrices, one doesn't need strings
constructed from the stack of powers of $A$, lhs of (\ref{eq:A=00003DVxE-1}),
to determine the eigenspace projector pattern. For a high enough power,
say $n$, the pattern $\Pi\left(A^{n}\right)$ is identical to the
eigenspace projector pattern.
\begin{thm}
\label{thm:ESPP}For real, symmetric positive definite matrix $A$,
there exists a finite integer $n$ such that the pattern in $A^{n}$
is identical to the eigenspace projector pattern, and all higher powers,
$A^{n+l}$, $l=1,\ldots$ have the same pattern.
\end{thm}

\begin{proof}
Assume symmetric positive definite matrix $A$ has $k$ distinct eigenvalues
and they are ordered $0<\lambda_{1}<\lambda_{2}<\cdots<\lambda_{k}$.
We only need to focus on a subset of locations from the lhs and rhs
of (\ref{eq:A=00003DVxE-1}), one for each distinct string. Let $(i,j)$
and $(r,s)$ be two off-diagonal locations with distinct strings.
Let $eij=E\left(:,i,j\right)$ and $ers=E\left(:,r,s\right)$ be the
column vectors associated with the strings on the rhs of (\ref{eq:A=00003DVxE-1})
and $aij$ and $ars$ be the column vectors associated with the strings
on the lhs of (\ref{eq:A=00003DVxE-1}). Note that 
\[
\left[\begin{array}{ccc}
1 & \cdots & 1\end{array}\right]\times eij=\left[\begin{array}{ccc}
1 & \cdots & 1\end{array}\right]\times ers=0
\]
 from (\ref{eq:SpectralDecomp}) since $(i,j)$ and $(r,s)$ are off-diagonal
locations.

Let $\delta=eij-ers$, then $\delta\neq0$ since the strings differ
but 
\begin{equation}
\left[\begin{array}{ccc}
1 & \cdots & 1\end{array}\right]\times\delta=0.\label{eq:sumDeltaEqualsZero}
\end{equation}
 Equation \ref{eq:sumDeltaEqualsZero} also applies to differences
between distinct strings on the diagonal since 
\[
\left[\begin{array}{ccc}
1 & \cdots & 1\end{array}\right]\times eii=\left[\begin{array}{ccc}
1 & \cdots & 1\end{array}\right]\times ejj=1.
\]
 Comparisons between diagonal and off-diagonals strings are ignored
since they can never match.

Given that $eij\neq ers$, then $aij\neq ars$ and there is at least
one power $t$, $1\leq t\leq k-1$, such that $A_{i,j}^{t}\neq A_{r,s}^{t}$.
The question remains, is there a finite integer $n$ such that $A_{i,j}^{n}\neq A_{r,s}^{n}$
and $A_{i,j}^{n+l}\neq A_{r,s}^{n+l}$ for all $l=1,...$.

Assume for every sequence of consecutive powers $t$ where $A_{i,j}^{t}-A_{r,s}^{t}\neq0$,
there is a maximum integer $n$ such that $A_{i,j}^{n}-A_{r,s}^{n}\neq0$
but $A_{i,j}^{n+1}-A_{r,s}^{n+1}=0$. Now write $A_{i,j}^{n+1}-A_{r,s}^{n+1}=0$
as 
\begin{equation}
A_{i,j}^{n+1}-A_{r,s}^{n+1}=\left[\lambda_{1}^{n+1},\cdots,\lambda_{k}^{n+1}\right]\times\delta=\left[\begin{array}{ccc}
1 & \cdots & 1\end{array}\right]\times\left[\begin{array}{c}
\lambda_{1}^{n+1}\delta_{1}\\
\vdots\\
\lambda_{k-1}^{n+1}\delta_{k-1}\\
\lambda_{k}^{n+1}\delta_{k}
\end{array}\right]=0.\label{eq:SumN+1EqualsZero}
\end{equation}
 Without loss of generality, assume that $\delta_{k}\neq0$, and make
the following substitution 
\begin{equation}
\left[\begin{array}{c}
\lambda_{1}^{n+1}\delta_{1}\\
\vdots\\
\lambda_{k-1}^{n+1}\delta_{k-1}\\
\lambda_{k}^{n+1}\delta_{k}
\end{array}\right]=\left[\begin{array}{c}
\lambda_{1}^{n+1}\delta_{1}\\
\vdots\\
\lambda_{k-1}^{n+1}\delta_{k-1}\\
-\left(\lambda_{1}^{n+1}\delta_{1}+\cdots+\lambda_{k-1}^{n+1}\delta_{k-1}\right)
\end{array}\right]\label{eq:BigSub}
\end{equation}
 where the right hand vector uses (\ref{eq:SumN+1EqualsZero}) to
substitute for $\lambda_{k}^{n+1}\delta_{k}$. Equating the last terms
of (\ref{eq:BigSub}) yields
\[
\lambda_{k}^{n+1}\delta_{k}=-\left(\lambda_{1}^{n+1}\delta_{1}+\cdots+\lambda_{k-1}^{n+1}\delta_{k-1}\right)
\]
or 
\[
\lambda_{k}^{n+1}=\frac{-\left(\lambda_{1}^{n+1}\delta_{1}+\cdots+\lambda_{k-1}^{n+1}\delta_{k-1}\right)}{\delta_{k}}.
\]
 Without loss of generality, assume $\delta_{k-1}\neq0$, and factor
$\lambda_{k-1}^{n+1}$ out of the numerator 
\[
\lambda_{k}^{n+1}=\lambda_{k-1}^{n+1}\left(\frac{-\left(\left(\frac{\lambda_{1}}{\lambda_{k-1}}\right)^{n+1}\delta_{1}+\cdots+1\delta_{k-1}\right)}{\delta_{k}}\right),
\]
and rearrange to get 
\begin{eqnarray}
\left(\frac{\lambda_{k}}{\lambda_{k-1}}\right)^{n+1} & = & \frac{-\left(\left(\frac{\lambda_{1}}{\lambda_{k-1}}\right)^{n+1}\delta_{1}+\cdots+1\delta_{k-1}\right)}{\delta_{k}}.\label{eq:contradiction}
\end{eqnarray}
 Taking the limit of both sides of (\ref{eq:contradiction}) results
in 
\[
\lim_{n\rightarrow\infty}\left(\frac{\lambda_{k}}{\lambda_{k-1}}\right)^{n+1}=+\infty\mbox{ while }\lim_{n\rightarrow\infty}\left(\frac{-\left(\left(\frac{\lambda_{1}}{\lambda_{k-1}}\right)^{n+1}\delta_{1}+\cdots+1\delta_{k-1}\right)}{\delta_{k}}\right)=\frac{-\delta_{k-1}}{\delta_{k}}
\]
 creating a contradiction. Therefore, there must be some finite integer
$n$ such that $A_{i,j}^{n}\neq A_{r,s}^{n}$ and $A_{i,j}^{n+l}\neq A_{r,s}^{n+l}$
for all $l=1,\ldots$.

Applying the same argument pairwise to all distinct off-diagonal strings
shows that a $n$ exists for each distinct pair of strings. Taking
the maximum over all of the $n$'s yields the desired result.
\end{proof}
Theorem \ref{thm:ESPP} does not guarantee that $n$ is less than
or equal to $k-1$, or that the value for $k$ is known. It states
that mathematically the pattern in powers of a real SPD matrix converges
to the eigenspace projector pattern. In practice, computer arithmetic
is finite precision. So if $n$ is large, the emergent pattern will
only reflect the eigenspace projector(s) associated with the dominant
eigenvalue(s).

\subsubsection{\label{subsec:Why-Symbolic-Squaring}Why Symbolic Matrix Multiplication}

Theorem \ref{thm:ESPP} says the eigenspace projector patterns of
integer SPD matrices can be computed using integer arithmetic. However,
it is not known when the eigenspace projector pattern will be reached,
and the computation is susceptible to overflow. When the author attempted
to find eigenspace projector patterns using recursive squaring, overflow
would occur after a few iterations. To avoid overflow, a symbol substitution
is performed to reduce the magnitude of the integers. The process
repeated alternating between recursive squaring and symbol substitution
until either the pattern converged or the number of symbols was so
large that symbol substitution could not reduce the magnitudes enough
to prevent overflow.

\textbf{Key observation}: There exist cases where the symbol substitution
resulted in a matrix whose eigenspace projector pattern is a strict
refinement of the original eigenspace projector pattern.\emph{ }These
``new'' matrices have the same p-similarity relationship as the
original matrices but a more refined eigenspace projector pattern.

An example is the first graph of case 10 from the website http://funkybee.narod.ru/graphs.htm.
The graph has 11 vertices and its adjacency matrix is given in Figure
\ref{fig:Case10Graph1FunkyBee}.
\begin{figure}
\begin{centering}
\includegraphics[scale=0.75]{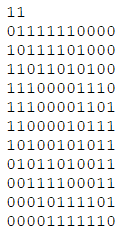}
\par\end{centering}
\caption{\label{fig:Case10Graph1FunkyBee}Case 10, graph 1 from http://funkybee.narod.ru/graphs.htm}

\end{figure}
 Its eigenspace projector pattern has 29 cells but recursive squaring
with symbol substitution resulted in a pattern with 30 cells. An exhaustive
search revealed four automorphisms. Applying the automorphisms yielded
30 orbits. The orbits are identical to the pattern generated by the
recursive squaring with symbol substitution. Comparing the orbits
with the eigenspace projector pattern shows that the cell associated
with symbol `6' in the eigenspace projector pattern contains two orbits,
see Figure \ref{fig:Case10Graph1_ESPP+Orbits}.
\begin{figure}
\centering{}\subfloat[Eigenspace Projector Pattern]{\includegraphics[width=2.5in]{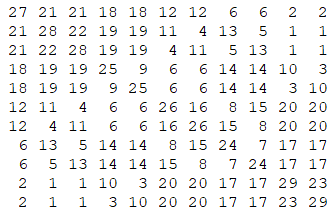}

}\hfill{}\subfloat[Orbits]{\includegraphics[width=2.5in]{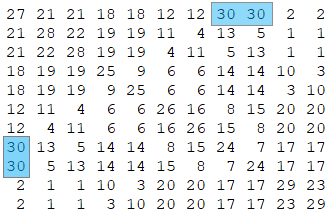}

}\caption{\label{fig:Case10Graph1_ESPP+Orbits}Case10 graph 1, eigenspace projector
pattern versus orbits}
\end{figure}

The question is how to find these more refined eigenspace projector
patterns without resorting to random symbol substitutions. This led
to examining symbolic matrix multiplication as a possible method for
simultaneously checking all possible symbol substitutions while avoiding
overflow.

\subsubsection{\label{subsec:CIPS}Canonical Inner Product Strings}

It turns out symbolic matrix multiplication also has an advantage
over regular matrix multiplication in separating orbits. Assume $A\underset{P}{\sim}B$
then there exists a $P$ such that $PAP^{T}=B$. If $P:i\rightarrow r\mbox{ and }j\rightarrow s$
then $A_{i,j}=B_{r,s}$. Also, 
\begin{eqnarray*}
PAP^{T}\times PAP^{T} & = & B\times B\\
P\left(AP^{T}\times PA\right)P^{T} & = & B\times B
\end{eqnarray*}
 implying 
\[
\left[AP^{T}\times PA\right]_{i,j}=\left[B\times B\right]_{r,s}
\]
where 
\[
A_{i,:}P^{T}\times PA_{:,j}=B_{r,:}\times B_{:,s}.
\]
 Since symmetric application of $P$ moves row $i$ of $A$ to become
row $r$ of $B$ (with the values permuted) and column $j$ of $A$
becomes column $s$ of $B$ (with the values permuted in the same
way) the multiset of terms in the inner products are identical. Further,
the order of the factors within a term are also identical. To see
why this is interesting, assume the inner products from two locations
are given by 
\[
\left[\begin{array}{cc}
\alpha & \beta\end{array}\right]\times\left[\begin{array}{c}
\beta\\
\alpha
\end{array}\right]=\alpha\beta+\beta\alpha=\gamma
\]
and 
\[
\left[\begin{array}{cc}
\alpha & \alpha\end{array}\right]\times\left[\begin{array}{c}
\beta\\
\beta
\end{array}\right]=\alpha\beta+\alpha\beta=\gamma.
\]
 Using regular matrix multiplication, both inner products appear to
be identical, evaluating to $\gamma$. However, symbolically the strings
formed from the vector of terms 
\[
\left[\begin{array}{c}
\alpha\beta\\
\beta\alpha
\end{array}\right]\mbox{ and \ensuremath{\left[\begin{array}{c}
\alpha\beta\\
\alpha\beta
\end{array}\right]}}
\]
are different, implying the two locations are not in the same orbit.
So symbolic inner products can separate orbits that appear to be the
same to regular matrix multiplication.
\begin{defn}
Define the \emph{canonical inner product string} of $u^{T}\times v$
as the ordered multiset of inner product terms where the first factor
in each term is from the row vector $u^{T}$ and the second from the
column vector $v$ 
\[
canonicalString\left(u^{T}\times v\right)=order\left(\left\{ \left\{ \left(u_{k},v_{k}\right)\right\} \right\} \right).
\]
To create the canonical inner product string, the multiset of terms
is split into two parts: terms involving diagonal symbols and terms
involving two off-diagonal symbols. For terms involving diagonal symbols,
they are ordered such that the row vector diagonal symbol comes first,
followed by the term involving the column diagonal symbol. For terms
involving two off-diagonal symbols, they are lexicographically ordered.
Then the ordered parts are concatenated, diagonal terms followed by
off-diagonal terms, to construct the canonical inner product string.
\end{defn}

\subsubsection{\label{subsec:symSqrDef} $\mathbf{SymSqr}$: Line 17}

A PCM is a symmetric positive integer matrix whose diagonal symbols
are distinct from off-diagonal symbols. Therefore the output of symbolically
squaring a PCM has a pattern that is symmetric and the diagonal canonical
inner product strings are distinct from off-diagonal canonical inner
product strings. The canonical inner product strings are constructed
as described in Section \ref{subsec:CIPS}. 
\begin{defn}
$SymSqr\left(\right)$ is a function that takes a symmetric matrix
whose diagonal symbols are distinct from off-diagonal symbols as input
and generates an array of canonical inner product strings as output.
Since the canonical inner product strings at $(i,j)$ and $(j,i)$
likely differ, the output array is made symmetric by choosing the
lessor of the canonical inner product strings at $(i,j)$ and $(j,i)$
to represent both locations. 
\end{defn}

Using a Jordan product to construct matching canonical strings for
$(i,j)$ and $(j,i)$ was considered. If a single string is constructed
from the combined terms of the inner products at $(i,j)$ and $(j,i)$,
it is possible that other locations $(r,s)$ and $(s,r)$, with distinct
canonical inner product strings, will generate the same string if
their terms are mixed together. This would hide the fact that the
locations are actually in different orbits. However, if the Jordan
product is interpreted as lexicographically ordering the pair of canonical
strings from $(i,j)$ and $(j,i)$ and concatenating them into a single
string, the final strings will be distinct and equivalent to choosing
the lessor string.
\begin{thm}
\label{thm:MatchingStrings}Locations in the same orbit have identical
canonical inner product strings.
\end{thm}

\begin{proof}
By the arguments used in Section \ref{subsec:CIPS} to derive and
construct the canonical inner product strings.
\end{proof}
The contrapositive of Theorem \ref{thm:MatchingStrings} causes patterns
to be refined. If two locations in the same cell of $\Pi\left(M\right)$
have distinct canonical strings when $M$ is symbolically squared,
they cannot be in the same orbit and will be assigned to different
cells in the new pattern $\Pi\left(SymSqr\left(M\right)\right)$.
\begin{thm}
\label{thm:DistinctStrings}Locations whose canonical inner product
strings differ are not in the same orbit.
\end{thm}

\begin{proof}
Contrapositive of Theorem \ref{thm:MatchingStrings}.
\end{proof}
\begin{thm}
\label{thm:NoSplittingOrbits}If $M$ is a square matrix whose diagonal
symbols are distinct from off-diagonal symbols, then the cells of
$\Pi\left(SymSqr\left(M\right)\right)$ represent disjoint sets of
orbits.
\end{thm}

\begin{proof}
This is a direct consequence of Theorem \ref{thm:MatchingStrings}.
Canonical inner product strings for locations in the same orbit are
identical. Therefore locations in the same orbit are assigned to the
same cell.
\end{proof}
Theorem \ref{thm:NoSplittingOrbits} does not guarantee there is one
orbit per cell. Only that locations in an orbit will not be split
across multiple cells.

The symbolic squaring function, $SymSqr\left(\right)$, appears twice
on line 17 of Algorithm \ref{alg:Blind-p-similarity-algorithm}. Each
input is a $m^{2}\times m^{2}$ symmetric, positive integer matrix
whose diagonal symbols are distinct from the off-diagonal symbols.
The outputs are $m^{2}\times m^{2}$ square arrays of canonical inner
product strings, where each string is composed of $m^{2}$ terms.

\subsubsection{\label{subsec:Monotonic-Refinement}Monotonic Refinement}

Let $M$ be a square matrix whose diagonal symbols are distinct from
the off-diagonal symbols. The next theorem addresses repeated symbolic
squaring of $M$. Theorem \ref{thm:MonotonicConvergence} shows that
symbolic squaring monotonically refines the pattern since elements
of two different cells will never be assigned to the same cell.
\begin{thm}
\label{thm:MonotonicConvergence}Repeated symbolic squaring of a square
matrix whose diagonal symbols are distinct from off-diagonal symbols
monotonically refines the pattern until a stable pattern is reached.
\end{thm}

\begin{proof}
Let $M$ be a square matrix whose diagonal symbols are distinct from
the off-diagonal symbols. To show that refinement is monotonic, it
is enough to show that two locations with distinct symbols can never
have identical canonical inner product strings. To see this, let $(i,j)$
and $(r,s)$ be two off-diagonal locations where $M_{i,j}\neq M_{r,s}$.
Without loss of generality assume that $i<j$ and $r<s$. Let $D$
be the diagonal of $M$. Then the inner products $\left[M\times M\right]_{i,j}$
and $\left[M\times M\right]_{r,s}$ look like 
\[
\left[M\times M\right]_{i,j}=\left[\begin{array}{ccccc}
\cdots & D_{i,i} & \cdots & M_{i,j} & \cdots\end{array}\right]\times\left[\begin{array}{c}
\vdots\\
M_{i,j}\\
\vdots\\
D_{j,j}\\
\vdots
\end{array}\right]=\left[\begin{array}{ccc}
1 & \cdots & 1\end{array}\right]\times\left[\begin{array}{c}
\vdots\\
D_{i,i}M_{i,j}\\
\vdots\\
M_{i,j}D_{j,j}\\
\vdots
\end{array}\right]
\]
 and 
\[
\left[M\times M\right]_{r,s}=\left[\begin{array}{ccccc}
\cdots & D_{r,r} & \cdots & M_{r,s} & \cdots\end{array}\right]\times\left[\begin{array}{c}
\vdots\\
M_{r,s}\\
\vdots\\
D_{s,s}\\
\vdots
\end{array}\right]=\left[\begin{array}{ccc}
1 & \cdots & 1\end{array}\right]\times\left[\begin{array}{c}
\vdots\\
D_{r,r}M_{r,s}\\
\vdots\\
M_{r,s}D_{s,s}\\
\vdots
\end{array}\right].
\]
 The canonical inner product string for $\left[M\times M\right]_{i,j}$
includes terms $\left\{ (D_{i,i},M_{i,j}),\,(M_{i,j},D_{j,j})\right\} $
and the canonical string for $\left[M\times M\right]_{r,s}$ includes
terms $\left\{ (D_{r,r},M_{r,s}),\,(M_{r,s},D_{s,s})\right\} $. Since
$M_{i,j}\neq M_{r,s}$ the strings do not match, and they will be
assigned distinct symbols.

A similar argument holds for diagonal inner product strings where
$M_{i,i}\neq M_{r,r}$. Their canonical inner product strings have
terms $\left(D_{i,i},D_{i,i}\right)$ and $\left(D_{r,r},D_{r,r}\right)$
respectively that do not match. Diagonal locations and off-diagonal
locations have a different structure to their canonical inner product
strings. Diagonal locations have a single term involving diagonal
symbols, whereas off-diagonal locations have two terms involving diagonal
symbols. Therefore, symbolic squaring either strictly refines the
pattern or has reached a stable pattern.
\end{proof}

\begin{thm}
\label{thm:MaxItersSymSqr}If $M$ is a symmetric $m\times m$ matrix
whose diagonal symbols are distinct from off-diagonal symbols, then
symbolic squaring followed by symbol substitution converges to a stable
pattern within $m$ iterations.
\end{thm}

\begin{proof}
Given $M$, a symmetric $m\times m$ matrix whose diagonal symbols
are distinct from off-diagonal symbols. Recall that Theorems \ref{thm:ConsistentSymbolSubstitution}
and \ref{thm:ShiftAndTranslate} imply that any pattern that is symmetric
and has diagonal values distinct from off-diagonal values can be replaced
by a SPD matrix with positive integer values. So, $M$ and the output
of symbolically squaring $M$ can always be made to be SPD matrices
with positive integer values.

Now consider equation (\ref{eq:A=00003DVxE-1}). The eigenspace projector
pattern on the rhs of (\ref{eq:A=00003DVxE-1}) can be generated from
the stack of $k-1$ powers of $M$ on the lhs, where $k$ is the number
of distinct eigenvalues. A $m\times m$ matrix has at most $m$ distinct
eigenvalues. So at most $m-1$ powers of $M$ are needed on the lhs
of (\ref{eq:A=00003DVxE-1}) to construct the eigenspace projector
pattern.

Next consider replacing successive powers of $M$ by symbolic squaring
with symbol substitution. At most $m-2$ symbolic squarings are needed
to fill the stack since $I$ and $M$ are the first two layers. Theorem
\ref{thm:MonotonicConvergence} guarantees that repeated symbolic
squaring monotonically refinements the pattern. Therefore, each successive
layer is a refinement of all prior layers. 

Theorem \ref{thm:ESPP} states that the powers of a SPD matrix converge
to the eigenspace projector pattern. The bottom layer of the stack
constructed from repeated symbolic squaring is a refinement of all
prior layers. Therefore the bottom layer is at least as refined as
the eigenspace projector pattern of $M$.

So symbolic squaring converges to a stable pattern within $m$ iterations.
\end{proof}
Symbolic squaring followed by symbol substitution occurs on line 17
of Algorithm \ref{alg:Blind-p-similarity-algorithm}.
\begin{defn}
Let $M$ be a square matrix whose diagonal symbols differ from off-diagonal
symbols. Define $\Pi\left(M^{*}\right)$ as the stable pattern reached
by repeated symbolic squaring and symbol substitution.
\end{defn}

One might wonder why symbolic squaring does not converge in a single
iteration. The answer appears to be related to the number of initial
symbols and the off-diagonal structure of a PCM. For the first iteration,
a PCM has at least five distinct symbols and the off-diagonal is highly
structured. For the five symbol case, see (\ref{eq:Blah}), it can
be shown that after the first iteration, there are 9 distinct symbols.
A second symbolic squaring results in 11 distinct symbols and is the
stable pattern. For another example, the PCM of a Petersen graph has
six symbols. The stable pattern has 65 symbols and is reached on the
third symbolic squaring. So the structure and limited number of initial
symbols typically causes symbolic squaring to take several iterations
to converge. In the authors experience, no case has taken more than
six iterations to converge and the number of distinct symbols can
be in the 10's of millions.

\subsection{\label{subsec:Comparing-Multisets}Comparing Multisets: Line 18}

Theorem \ref{thm:MatchingStrings} shows that matching canonical inner
product strings is a necessary condition for two locations to be in
the same orbit. This implies the multiset of canonical inner product
strings for $SymSqr\left(A\right)$ and $SymSqr\left(B\right)$ must
match for $A$ and $B$ to be p-similar, $mix\left(SymSqr\left(A\right)\right)=mix\left(SymSqr\left(B\right)\right)$\footnote{See Section \ref{subsec:Symmetric-Permutation} for the definitions
of the various mixes.}. Conversely, if the mixes do not match, then $A$ and $B$ cannot
be p-similar. However it is possible to detect when the mixes do not
match by looking at a smaller multiset. If the mixes do not match,
$mix\left(SymSqr\left(A\right)\right)\neq mix\left(SymSqr\left(B\right)\right)$,
then the column mixes do not match, $colMix\left(SymSqr\left(A\right)\right)\neq colMix\left(SymSqr\left(B\right)\right)$,
which in turn implies the diagonal mixes will not match, $diagMix\left(SymSqr\left(A\right)\right)\neq diagMix\left(SymSqr\left(B\right)\right)$.
Column mixes may differ one iteration before the diagonal mixes if
the differences are in off-diagonal symbols. However, after the next
iteration, the off-diagonal differences are reflected in the diagonal
symbols causing the diagonal mixes to differ. Comparing diagonal mixes
is performed on line 18 in Algorithm \ref{alg:Blind-p-similarity-algorithm}.

Comparing diagonal mixes to determine when two matrices are not p-similar
is consistent with Specht's Theorem \cite[pp. 97-98]{hornAndJohnson2013}.
Specht's Theorem gives necessary and sufficient conditions for showing
when two matrices are unitarily similar. Specht's Theorem works by
comparing the traces of sequences of matrix products of various lengths.
For p-similarity it is not enough to show that the traces match, the
diagonal mixes must match. So, matching diagonal mixes, after every
round of symbolic squaring, is a necessary condition for $A$ and
$B$ to be p-similar. This implies differing diagonal mixes is a sufficient
condition for declaring two matrices non p-similar. Note that symbolic
squaring with symbol substitution means the diagonal mixes of an uncountably
infinite number of matrix pairs are compared each iteration.

\subsection{Comparing Patterns, Line 20}

The last function in Algorithm \ref{alg:Blind-p-similarity-algorithm}
compares two patterns and returns TRUE if the patterns are identical.
If the pattern before, $\Pi\left(M\right)$, and after symbolic squaring,
$\Pi\left(SymSqr\left(M\right)\right)$, match then the stable pattern
$\Pi\left(M^{*}\right)$ has been reached for that matrix. Necessary
conditions for $A$ and $B$ to be p-similar are that the diagonal
mixes match and they reach stable patterns on the same iteration.
The first condition is checked on line 18 and the second on line 20
of Algorithm \ref{alg:Blind-p-similarity-algorithm}.

Note that the case where diagonal mixes match and one PCM has reached
a stable pattern but the other has not, will be detected as non p-similar.
The non stable pattern is refined by the symbolic squaring in the
next iteration, changing the column mixes. This in turn implies the
diagonal symbols will change, causing the diagonal mixes to differ.
Ultimately resulting in the PCMs being declared non p-similar.

\subsection{\label{subsec:Algorithm-Complexity}Algorithm \ref{alg:Blind-p-similarity-algorithm}
Complexity}

The overall complexity of Algorithm \ref{alg:Blind-p-similarity-algorithm}
for comparing two $m\times m$ matrices is $O\left(m^{12}\right)$.
The bound is conservative and computed using simple, unoptimized approaches.
In particular, symbol substitutions are assumed to use the method
described in Section \ref{subsec:Symbol-Substitution}.

Below are complexity bounds on individual functions of Algorithm \ref{alg:Blind-p-similarity-algorithm}
given two $m\times m$ matrices as input:
\begin{description}
\item [{Line~11,~SymbolSubstitution}] $O\left(m^{4}\right)$ since there
are $2m^{2}$ locations containing values and $O\left(m^{2}\right)$
comparison operations per substitution with a constant cost per comparison
operation.
\item [{Lines~13\&14,~ShiftAndTranslate}] $O\left(m^{2}\right)$ since
there are $2m^{2}$ locations to update.
\item [{Lines~13\&14,~PCM}] $O\left(m^{4}\right)$ since there are $2m^{4}$
locations to initialize.
\item [{Line~17,~SymSqr}] $O\left(m^{8}\right)$ reasoned as follows:
\begin{itemize}
\item Construct terms from factors: $O\left(m^{6}\right)$ since there are
$2m^{4}$ canonical strings and each string has $m^{2}$ terms and
a constant cost to construct a term from its factors.
\item Construct canonical inner product strings from the terms: $O\left(m^{8}\right)$
since there are $2m^{4}$ strings and it takes a maximum of $O\left(m^{4}\right)$
comparison operations to order the $m^{2}$ terms of a string at a
constant cost per comparison operation.
\item Select lessor string for $(i,j)$ and $(j,i)$ locations: $O\left(m^{6}\right)$
since there are $2m^{4}$ locations and comparing two strings takes
$O\left(m^{2}\right)$ comparison operations at a constant cost per
comparison operation.
\end{itemize}
\item [{Line~17,~SymbolSubstitution}] $O\left(m^{10}\right)$ since there
are $2m^{4}$ locations and $O\left(m^{4}\right)$ string comparison
operations per substitution at a cost of $O\left(m^{2}\right)$ to
compare two strings.
\item [{Line~18,~CompareDiagMultisets}] $O\left(m^{4}\right)$ since
it takes a maximum of $O\left(m^{4}\right)$ comparison operations
to sort a multiset with $m^{2}$ values at a constant cost per comparison
operation. Comparing the sorted multisets takes at most $O\left(m^{2}\right)$
comparison operations.
\item [{Line~20,~ComparePatterns}] $O\left(m^{8}\right)$ since the comparison
is performed by replacing values with their representative locations
and then comparing the results. To replace values with representative
locations involves $2m^{4}$ locations and $O\left(m^{4}\right)$
comparison operations per substitution at a constant cost per comparison
operation. This is followed by at most $O\left(m^{4}\right)$ comparison
operations to compare the two patterns. 
\item [{Lines~16-24,~Repeat~Until~Iterations}] $O\left(m^{2}\right)$
by Theorem \ref{thm:MaxItersSymSqr} applied to a $m^{2}\times m^{2}$
SPD matrix.
\item [{Overall~Complexity}] $O\left(m^{12}\right)$ derived from iterating
the repeat until loop a maximum of $O\left(m^{2}\right)$ times, and
a $O\left(m^{10}\right)$ cost to substitute symbols for canonical
inner product strings each iteration.
\end{description}
As mentioned previously, in the author's experience, no case has taken
more the six iterations to converge. The greatest reduction in overall
complexity will come from optimizing the symbolic substitution applied
to canonical inner product strings. Even better would be to find a
way to avoid symbolic matrix multiplication all together, see Section
\ref{subsec:MatMultIsMonotonic}. In terms of space, Algorithm \ref{alg:Blind-p-similarity-algorithm}
requires $O\left(m^{6}\right)$ locations to hold the $2m^{4}$ canonical
inner product strings, where a location is large enough to hold a
term.

\section{\label{sec:The-Crux}Sufficiency Argument}

In Section \ref{sec:Blind-P-similarity-Algorithm}, Algorithm \ref{alg:Blind-p-similarity-algorithm}
is shown to test necessary conditions for two PCMs to be p-similar
and that differing diagonal mixes are sufficient to determine that
two matrices are not p-similar. The last theoretical hole is to determine
whether the patterns of non p-similar PCMs can stabilize on the same
iteration with matching diagonal mixes. The purpose of this section
is to argue that it is not possible, implying Algorithm \ref{alg:Blind-p-similarity-algorithm}
is necessary and sufficient for determining when two matrices are
p-similar. 

An outline of the argument goes as follows: 
\begin{enumerate}
\item Given a $m^{2}\times m^{2}$ PCM and its automorphism group, the set
of all symmetric positive definite matrices with the same automorphism
group is constructed.
\item The set is shown to be convex. 
\item Repeated symbolic squaring with symbol substitution is shown to act
as a descent method on the convex set, implying the orbits are found
within $m^{2}$ steps. 
\item Having the orbits does not guarantee the symbol to orbit assignments
are consistent across matrices. 
\item Using properties of PCMs, it is shown that an Oracle either matches
all of the columns $S^{*}$ and $T^{*}$, or none of the columns.
This differs from cases where an Oracle is able to partially match
non p-similar matrices, Section \ref{subsec:Origin-of-PCGs}.
\item Next, the PCM of the color matrix for the direct sum of the initial
pair of matrices, $M=PCM\left(\left(A\oplus B\right)_{C}\right)$,
is constructed and symbolically squared until its orbits, $\Pi\left(M^{*}\right)$,
are found. PCMs $S$ and $T$ are embedded within PCM $M$.
\item The patterns at locations associated with $S$ and $T$ in $\Pi\left(M^{*}\right)$
must match $\Pi\left(S^{*}\right)$and $\Pi\left(T^{*}\right)$ since
they are the orbits imposed by $Aut\left(S\right)$ and $Aut\left(T\right)$,
respectively. 
\item Since $S$ and $T$ are not p-similar, the Oracle guaranteed no column
of $S^{*}$ can be matched with any column of $T^{*}$. So orbits
associated with locations of $S^{*}$ in $M^{*}$ are distinct from
orbits associated with locations of $T^{*}$ in $M^{*}$. 
\item So the diagonal symbols in $M^{*}$ associated with locations of $S^{*}$
and $T^{*}$ are distinct. 
\item Therefore, the diagonal symbols of $S^{*}$ and $T^{*}$ are distinct
as long as consistent symbol substitution is performed each symbolic
squaring.
\item This implies comparing diagonal mixes of $S^{*}$ and $T^{*}$ is
sufficient for detecting non p-similar matrices.
\end{enumerate}

\subsection{Crux of the Argument}

Showing that repeated symbolic squaring of a PCM converges to the
orbits is the crux of the argument. Given that it is true, a simplified
argument constructs $M=PCM\left(\left(A\oplus B\right)_{C}\right)$
and then either the multisets of diagonal locations associated with
$A_{C}$ and $B_{C}$ in $M^{*}$ match or they differ, implying one
or more of the orbits are distinct. This would work as a blind permutation
similarity algorithm. This follows the graph isomorphism argument
that two graphs are isomorphic iff there exists an automorphism of
their direct sum that exchanges the graphs.

Repeated symbolic squaring of the direct sum $\left(A\oplus B\right)$,
without constructing the PCM of the direct sum, is not guaranteed
to separate non p-similar matrices. For example, the non p-similar
adjacency matrices `had-sw-32-1' and `had-sw-32-2' from the Bliss
graph collection \cite{bliss} do not have their orbits separated
by repeated squaring of the direct sum. However, their associated
PCM's $S$ and $T$ have distinct diagonal mixes after the third symbolic
squaring.

Algorithm \ref{alg:Blind-p-similarity-algorithm} uses $S$ and $T$
instead of the PCM of the direct sum $M$ because $S$ and $T$ are
four times smaller. This makes working with $S$ and $T$ more practical
than working with $M$. 

\subsection{\label{subsec:WLOG-SPD-Matrices}WLOG Symbolic Squaring Generates
SPD Matrices}

As pointed out in Section \ref{subsec:symSqrDef}, inputs to the symbolic
squaring function $SymSqr()$ are symmetric, positive integer matrices
whose diagonal symbols are distinct from off-diagonal symbols. The
output is a symmetric array of canonical inner product strings whose
diagonal strings are distinct from off-diagonal strings. This is immediately
followed by a symbol substitution where the canonical inner product
strings are replaced by positive integers resulting in an output that
is a symmetric, positive integer matrix whose diagonal symbols are
distinct from off-diagonal symbols. So, one is free to choose the
symbol set, as long as they are positive integers. 

In particular, assume there are $n$ distinct canonical inner product
strings, of which there are $n_{1}$ distinct off-diagonal strings
and $n_{2}$ distinct diagonal strings, with $n_{1}+n_{2}=n$. Choose
the consecutive integers $1$ to $n_{1}$ as the set of symbols to
replace off-diagonal strings and consecutive integers $\left(m^{2}n_{1}+1\right)$
to $\left(m^{2}n_{1}+n_{2}\right)$ as the set of symbols to replace
diagonal strings. Then the resulting matrix will be symmetric and
diagonally dominant by construction, and so symmetric positive definite
(SPD). Therefore in the sufficiency arguments below, symbolic squaring
followed by symbol substitution generates SPD matrices unless otherwise
noted.

\subsection{\label{subsec:SymSub}${\bf SymSub}$}

Algorithm \ref{alg:Blind-p-similarity-algorithm} uses $SymbolSubstitution$
to perform a consistent symbol substitution on two arrays of canonical
inner product strings. In this section, symbol substitution is often
performed on a single array of canonical inner product strings using
a function called $SymSub$. The substitutions are performed in a
permutation independent fashion. The set of $n_{1}$ distinct off-diagonal
strings are lexicographically ordered and then 1 is assigned to the
first string in the ordered list and so on until $n_{1}$ is assigned
to the last string in the ordered list. Similarly, the list of $n_{2}$
distinct diagonal strings are ordered prior to assigning $\left(m^{2}n_{1}+1\right)$
to $\left(m^{2}n_{1}+n_{2}\right)$ to them. This guarantees that
any symmetric permutation applied to the array of canonical inner
product strings will assign the same integer to the same string and
that the result is SPD.

\subsection{A Convex Set}

Theorem \ref{thm:ABvsST} shows that determining p-similarity between
two square matrices is equivalent to determining the p-similarity
between two PCMs. PCMs are symmetric matrices whose diagonal symbols
differ from off-diagonal symbols. Since the spectrum can be shifted
without affecting the p-similarity relationship, there exist SPD PCMs
with the same pattern. One way to generate a SPD PCM is to re-define
the color matrix to use $\gamma=2m^{2}$ instead of $\gamma=2$. Using
$\gamma=2m^{2}$ guarantees the PCM is SPD since it is symmetric and
diagonally dominant by construction. 
\begin{defn}
Given a $m^{2}\times m^{2}$ SPD PCM $M$, let $\mathbb{M}=\left\{ N\in SPD\,|\,Aut\left(N\right)=Aut\left(M\right)\right\} $
be the set of all real SPD matrices whose automorphism group is $Aut\left(M\right)$.
\end{defn}

\begin{thm}
\label{thm:AutIsConvex} $\mathbb{M}$ is a convex set. That is for
all $M1$, $M2\in\mathbb{M}$, the linear combination $\alpha M1+\left(1-\alpha\right)M2$
is also in $\mathbb{M}$ for $0\leq\alpha\leq1$.
\end{thm}

\begin{proof}
Let $M1$ and $M2$ be two matrices in $\mathbb{M}$. It is well known
that the set of real SPD matrices is convex, so the focus of the proof
is on showing that $Aut\left(\alpha M1+\left(1-\alpha\right)M2\right)=Aut\left(M\right)$
for $0\leq\alpha\leq1$.

$\Rightarrow$ For all $P\in Aut\left(M\right)$, $PM1P^{T}=M1$ and
$PM2P^{T}=M2$ since $M1$ and $M2$ are in $\mathbb{M}$. So we have
\begin{eqnarray*}
\alpha PM1P^{T}+\left(1-\alpha\right)PM2P^{T} & = & \alpha M1+\left(1-\alpha\right)M2\\
P\left(\alpha M1+\left(1-\alpha\right)M2\right)P^{T} & = & \alpha M1+\left(1-\alpha\right)M2
\end{eqnarray*}
 showing that $Aut\left(M\right)\leq Aut\left(\alpha M1+\left(1-\alpha\right)M2\right)$.

$\Leftarrow$ For all $P\in Aut\left(\alpha M1+\left(1-\alpha\right)M2\right)$,
we have 
\begin{equation}
P\left(\alpha M1+\left(1-\alpha\right)M2\right)P^{T}=\alpha M1+\left(1-\alpha\right)M2\mbox{ for }0\leq\alpha\leq1.\label{eq:ConvexAutGroup}
\end{equation}
\begin{casenv}
\item For $\alpha=0$, (\ref{eq:ConvexAutGroup}) equals $PM2P^{T}=M2$
so $Aut\left(\alpha M1+\left(1-\alpha\right)M2\right)\leq Aut\left(M\right)$
for $\alpha=0$.
\item For $0<\alpha\leq1$, (\ref{eq:ConvexAutGroup}) can be rewritten
as 
\[
PM1P^{T}+\gamma PM2P^{T}=M1+\gamma M2
\]
 where $\gamma=\frac{\left(1-\alpha\right)}{\alpha}\geq0$. Rearrange
to get 
\begin{equation}
PM1P^{T}-M1=\gamma\left(M2-PM2P^{T}\right).\label{eq:ConvexAutRearranged}
\end{equation}
 Now the matrix differences on the lhs and rhs of (\ref{eq:ConvexAutRearranged})
are fixed for a given $P$ and yet (\ref{eq:ConvexAutRearranged})
holds for all $\gamma\geq0$. This creates a contradiction unless
$\left(PM1P^{T}-M1\right)$ and $\left(M2-PM2P^{T}\right)$ are zero
matrices, implying $PM1P^{T}=M1$ and $PM2P^{T}=M2$. So $Aut\left(\alpha M1+\left(1-\alpha\right)M2\right)\leq Aut\left(M\right)$
for $0<\alpha\leq1$.
\end{casenv}
Therefore $Aut\left(\alpha M1+\left(1-\alpha\right)M2\right)=Aut\left(M\right)$
for $0\leq\alpha\leq1$ so each $\alpha M1+\left(1-\alpha\right)M2$
is in $\mathbb{M}$ implying $\mathbb{M}$ is a convex set.
\end{proof}
$\mathbb{M}$ is not closed under matrix multiplication. For any $N\in\mathbb{M}$,
its inverse $N^{-1}$ is also in $\mathbb{M}$, but the automorphism
group of $N\times N^{-1}=I$ is $S_{m}$, the symmetric group, which
is not a subgroup of $\left\{ P\otimes P\right\} $. However, the
powers of a matrix in $\mathbb{M}$ are in $\mathbb{M}$.

\subsection{$\mathbf{M_{Aut}}$ and $\mathbf{M^{*}}$ are in $\mathbf{\mathbb{M}}$}

Matrices in $\mathbb{M}$ can be separated into equivalence classes
by their patterns (partitions). Elements of a class differ from each
other by a symbol substitution. 
\begin{defn}
\label{def:Maut}Let $M$ be a SPD PCM. Define $M_{Aut}$ as a SPD
matrix representing the orbits imposed by $Aut\left(M\right)$ on
$L$, the set of $(i,j)$ locations, see Section \ref{subsec:Triplet-Notation}.
Locations in the same orbit are assigned the same symbol and locations
in distinct orbits have distinct symbols. Symbols on the diagonal
of $M_{Aut}$ are distinct from its off-diagonal symbols since the
orbits of diagonal elements are distinct from the orbits of off-diagonal
elements. $M_{Aut}$ is symmetric since $M$ is symmetric and the
diagonal symbols are chosen so that $M_{Aut}$ is SPD.
\end{defn}

$M_{Aut}$ is in $\mathbb{M}$ since it is SPD and $Aut\left(M_{Aut}\right)=Aut\left(M\right)$
by the definition of $M_{Aut}$.
\begin{defn}
\label{def:M*}Let $M$ be a SPD PCM. Let $M^{(i)}=SymSub\left(SymSqr\left(M^{(i-1)}\right)\right)$
be the result of the $i^{{\rm th}}$ symbolic squaring and symbol
substitution, where $M^{(0)}=M$. Let $s$ be the iteration where
the stable pattern is reached, then $\Pi\left(M^{(s+j)}\right)=\Pi\left(M^{(s)}\right)$
for $j=1,2,\ldots$. Define $M^{*}$ as a SPD matrix whose pattern
matches $M^{(s)}$, $\Pi\left(M^{*}\right)=\Pi\left(M^{(s)}\right)$.
Each $M^{(i)}$, $i=1,\ldots$, is SPD by Section \ref{subsec:WLOG-SPD-Matrices}.
\end{defn}

The next theorem shows that the $M^{(i)}$ and $M^{*}$ are in $\mathbb{M}$.
\begin{thm}
\label{thm:Aut(M*)=00003DAut(M)}Given SPD PCM $M$ and $\mathbb{M}$.
If $M^{(i)}=SymSub\left(SymSqr\left(M^{(i-1)}\right)\right)$ where
$M^{(0)}=M$ and $M^{*}$ is a SPD matrix representing the stable
pattern resulting from repeated symbolic squaring. Then the $M^{(i)}$,
for $i=1,\ldots$, and $M^{*}$ are in $\mathbb{M}$.
\end{thm}

\begin{proof}
Given PCM SPD $M$ and $\mathbb{M}$, the set of all SPD matrices
whose automorphism group is $Aut\left(M\right)$, to show that $M^{(i)}$
and $M^{*}$ are in $\mathbb{M}$ we need to show that they are SPD
and their automorphism groups match $Aut\left(M\right)$. Since $M^{(i)}$
and $M^{*}$ are SPD by Definition \ref{def:M*}, the focus is on
showing that their automorphism groups match $Aut\left(M\right)$.

Let $M^{(i)}=SymSub\left(SymSqr\left(M^{(i-1)}\right)\right)$ where
$M^{(0)}=M$. 

For $i=1$: $M^{(1)}=SymSub\left(SymSqr\left(M^{(0)}\right)\right)$.
To see that $Aut\left(M^{(1)}\right)=Aut\left(M\right)$, first note
that Theorem \ref{thm:MonotonicConvergence} guarantees that 
\[
\Pi\left(M^{(1)}\right)=\Pi\left(SymSub\left(SymSqr\left(M^{(0)}\right)\right)\right)\preceq\Pi\left(M^{(0)}\right)=\Pi\left(M\right).
\]

$\Rightarrow$ For each $P\in Aut\left(M^{(1)}\right)$, if $P$ symmetrically
permutes $(i,j)$ to $(r,s)$ then $M_{i,j}^{(1)}=M_{r,s}^{(1)}$
and $M_{i,j}=M_{r,s}$ since $\Pi\left(M^{(1)}\right)\preceq\Pi\left(M\right)$.
So $PMP^{T}=M$ implying $P\in Aut\left(M\right)$ and $Aut\left(M^{(1)}\right)\leq Aut\left(M\right)$. 

$\Leftarrow$ For each $P\in Aut\left(M\right)$, if $P$ symmetrically
permutes $(i,j)$ to $(r,s)$ then $(i,j)$ and $(r,s)$ are in the
same orbit and Theorem \ref{thm:MatchingStrings} guarantees both
locations have identical canonical inner product strings. Symbol substitution
assigns the same symbol to identical canonical inner product strings
resulting in $M_{i,j}^{(1)}=M_{r,s}^{(1)}$, so $PM^{(1)}P^{T}=M^{(1)}$
and $P\in Aut\left(M^{(1)}\right)$ implying $Aut\left(M\right)\leq Aut\left(M^{(1)}\right)$. 

Therefore $Aut\left(M^{(1)}\right)=Aut\left(M\right)$.

By induction, the automorphism group of $M^{(i)}$ is $Aut\left(M\right)$
for $i=2,\ldots$. 

For $M^{*}$, let $M^{*}=M^{(s)}$ where $s$ is the iteration a stable
pattern is reached, i.e., $\Pi\left(M^{(s+1)}\right)=\Pi\left(M^{(s)}\right)$.
Then $M^{*}$ is in $\mathbb{M}$ since $M^{(s)}$ is in $\mathbb{M}$. 

Therefore all of the intermediate $M^{(i)},$$i=1,\ldots$, and $M^{*}$
are in $\mathbb{M}$.
\end{proof}

\subsection{Lattice and Chain}

The patterns (partitions) of the matrices in $\mathbb{M}$ form equivalence
classes that can be partially ordered by refinement. Let $M_{sup}$
be a PCM in $\mathbb{M}$ with the fewest distinct symbols/cells.
Then the partial ordering forms a complete lattice with $\Pi\left(M_{sup}\right)$
is the supremum and $\Pi\left(M_{Aut}\right)$ as the infimum \cite[pg. 436]{jacobson74}. 

The subset of patterns $\left\{ \Pi\left(M\right),\,\Pi\left(M^{(1)}\right),\,\cdots,\,\Pi\left(M^{*}\right)\right\} $,
where
\[
\Pi\left(M^{*}\right)\prec\cdots\prec\Pi\left(M^{(1)}\right)\prec\Pi\left(M\right)
\]
 is a totally ordered set, called a chain, between $\Pi\left(M\right)$
and $\Pi\left(M^{*}\right)$ \cite{jacobson74}. 

Both $\Pi\left(M^{*}\right)$ and $\Pi\left(M_{Aut}\right)$ are fixed
point patterns (stable partitions) of symbolic squaring, 
\begin{eqnarray*}
\Pi\left(SymSqr\left(M^{*}\right)\right) & = & \Pi\left(M^{*}\right)
\end{eqnarray*}
 and 
\begin{eqnarray*}
\Pi\left(SymSqr\left(M_{Aut}\right)\right) & = & \Pi\left(M_{Aut}\right).
\end{eqnarray*}
Tarski's Lattice-Theoretical Fixed Point Theorem \cite{tarski1955}
guarantees that the set of fixed points of symbolic squaring is not
empty, and that the fixed points form a complete lattice partially
ordered by refinement. Unfortunately it does not say that $\Pi\left(M^{*}\right)$
equals $\Pi\left(M_{Aut}\right)$. If they are identical, then the
process of symbolic squaring PCMs finds the orbits.

\subsection{\label{subsec:ConjugateDirectionDescent}Symbolic Squaring as a Descent
Method}

Symbolic squaring with symbol substitution monotonically refines a
PCM $M$ until a stable pattern, $\Pi\left(M^{*}\right)$, is reached.
Cells of $\Pi\left(M^{(i)}\right)$ represent disjoint sets of orbits,
Theorem \ref{thm:NoSplittingOrbits}. $\mathbb{M}$ is a convex set
composed of all SPD matrices whose automorphism group matches $Aut\left(M\right)$.
Symbolic squaring creates a chain of strictly refined patterns from
$\Pi\left(M\right)$ to $\Pi\left(M^{*}\right)$. The question is
whether $\Pi\left(M^{*}\right)$ is equal to $\Pi\left(M_{Aut}\right)$.

The process of refinement can be viewed as a descent method. Given
$M$ and $\mathbb{M}$, for all $P\in Aut\left(M\right)$ one can
write 
\begin{equation}
M\times M-P\left(MP^{T}\times PM\right)P^{T}=0.\label{eq:AsMinimizationProblem}
\end{equation}
 Equation \ref{eq:AsMinimizationProblem} can be converted to a pattern
difference that is symbol independent as 
\begin{equation}
\left\Vert \Pi\left(M\times M\right)-\Pi\left(P\left(MP^{T}\times PM\right)P^{T}\right)\right\Vert _{F}^{2}=0\label{eq:PiMxM}
\end{equation}
 where $\left\Vert \right\Vert _{F}^{2}$ is the square of the Frobenius
norm and the pattern difference assigns zero to locations with matching
representative locations and one to locations where the representative
locations do not match, see Section \ref{subsec:Triplet-Notation}.
Convert (\ref{eq:PiMxM}) to using symbolic squaring with symbol substitution
as 
\[
\left\Vert \Pi\left(SymSub\left(SymSqr\left(M\right)\right)\right)-\Pi\left(P\left(SymSub\left(SymSqr\left(M\right)\right)\right)P^{T}\right)\right\Vert _{F}^{2}=0
\]
where the function $SymSub()$ performs symbol substitution on a single
matrix using a permutation independent mapping from symbol to canonical
inner product string, as described in Section \ref{subsec:SymSub}.

Let $M_{Aut}$ be a SPD matrix representing the orbits imposed by
$Aut\left(M\right)$, as defined in Definition \ref{def:Maut}. Let
$M^{(i)}=SymSub\left(SymSqr\left(M^{(i-1)}\right)\right)$ where $\Pi\left(M^{(0)}\right)=\Pi\left(M\right)$
and $M^{(i)}\in\mathbb{M}$. Then the objective is to minimize the
cost function 
\begin{equation}
\left\Vert \Pi\left(SymSub\left(SymSqr\left(M_{Aut}\right)\right)\right)-\Pi\left(P\left(SymSub\left(SymSqr\left(M^{(i-1)}\right)\right)\right)P^{T}\right)\right\Vert _{F}^{2}=SSE\label{eq:CostFunction}
\end{equation}
for $i=1,\cdots,*$ where SSE is the sum square error. Since $\Pi\left(M_{Aut}\right)$
is a stable pattern, and each $M^{(i)}$ is in $\mathbb{M}$, (\ref{eq:CostFunction})
can be rewritten as 
\begin{eqnarray}
\left\Vert \Pi\left(M_{Aut}\right)-\Pi\left(M^{(i)}\right)\right\Vert _{F}^{2} & = & SSE.\label{eq:FinalCostFunction}
\end{eqnarray}
 The SSE in (\ref{eq:FinalCostFunction}) decreases monotonically.
Each cell of $\Pi\left(M^{(i)}\right)$ that contains multiple orbits,
has only one orbit with the correct representative location, the remaining
orbits in the cell have the wrong representative location. As symbolic
squaring refines a cell, at least one additional orbit will have the
correct representative location. So the SSE of the pattern difference
decreases every iteration until the stable pattern $\Pi\left(M^{*}\right)$
is reached.

Although we don't know $\Pi\left(M_{Aut}\right)$ the monotonic refinement
of symbolic squaring along with the convexity of $\mathbb{M}$ drives
the iterations towards $\Pi\left(M_{Aut}\right)$. So an iterative
method given by 
\begin{equation}
M^{(i)}=SymSub\left(SymSqr\left(M^{(i-1)}\right)\right)\label{eq:Iterate}
\end{equation}
 that halts when 
\begin{equation}
\left\Vert \Pi\left(M^{(i)}\right)-\Pi\left(M^{(i-1)}\right)\right\Vert _{F}^{2}=0\label{eq:ConvergenceTest}
\end{equation}
causes the successive patterns $\Pi\left(M^{(i)}\right)$ to approach
$\Pi\left(M_{Aut}\right)$. Algorithm \ref{alg:Blind-p-similarity-algorithm}
uses a variant of (\ref{eq:ConvergenceTest}), namely testing $\Pi\left(M^{(i)}\right)$
and $\Pi\left(M^{(i-1)}\right)$ for equality, to detect when the
process has reached a stable pattern, line 20 of Algorithm \ref{alg:Blind-p-similarity-algorithm}. 

In summary, repeated symbolic squaring with symbol substitution finds
the most refined pattern. The refinement is monotonic and converges
to a stable pattern where locations in an orbit are in the same cell.
By the convexity of $\mathbb{M}$, $\Pi\left(M^{*}\right)$ equals
$\Pi\left(M_{Aut}\right)$. 

Given that symbolic squaring finds the orbits, one can construct PCM
$M=PCM\left(\left(A\oplus B\right)_{C}\right)$ from the color matrix
for the direct sum of the input matrices and check whether locations
on the diagonal of $M^{*}$ associated with $A$ and $B$ are in the
same orbits. If not, $A$ and $B$ are not p-similar. The sufficiency
argument would end here. However, the author chooses to continue the
sufficiency argument because a) it will be shown that not only do
the orbits differ for non p-similar matrices but no orbits match,
and b) comparing multisets of diagonal elements to determine p-similarity
fits nicely with Specht's Theorem comparing traces to determine unitary
similarity.

Although the process of symbolic squaring a PCM finds its orbits,
it says nothing about the association of symbol to orbit. Two PCMs,
$A$ and $B$, if processed separately can have the same automorphism
group $Aut\left(A\right)=Aut\left(B\right)$ and partitions $\Pi\left(A\right)=\Pi\left(B\right)$
so the orbits match $M_{Aut\left(A\right)}=M_{Aut\left(B\right)}$.
But if they have different symbol sets, $\Sigma_{A}\neq\Sigma_{B}$
or the same symbol set with different mappings of symbol to cell $g_{A}\neq g_{B}$,
then the matrices are not p-similar. This issue is addressed in the
following sections.

\subsection{${\bf SymMult}$}

Before addressing the issue of symbols to orbits, there is one more
symbolic matrix multiplication function that is useful. It is called
$SymMult\left(\right)$ and takes two diagonal matrices and a square
matrix as input. The first diagonal matrix is applied to the square
matrix from the left and the second from the right, $SymMult\left(D1,\,M,\,D2\right)$.
For example, if 
\[
D1=\left[\begin{array}{cc}
\delta_{1} & 0\\
0 & \delta_{2}
\end{array}\right]\mbox{, }M=\left[\begin{array}{cc}
M_{1,1} & M_{1,2}\\
M_{2,1} & M_{2,2}
\end{array}\right]\mbox{, and }D2=\left[\begin{array}{cc}
\overline{\delta}_{1} & 0\\
0 & \overline{\delta}_{2}
\end{array}\right]
\]
 then $SymMult\left(D1,M,D2\right)$ is the canonical inner product
string array 
\begin{equation}
SymMult\left(D,M,D\right)=\left[\begin{array}{cc}
\delta_{1}M_{1,1}\overline{\delta}_{1} & \delta_{1}M_{1,2}\overline{\delta}_{2}\\
\delta_{2}M_{2,1}\overline{\delta}_{1} & \delta_{2}M_{2,2}\overline{\delta}_{2}
\end{array}\right].\label{eq:SysMult(DMD)}
\end{equation}

\begin{thm}
\label{thm:P(DAD)P' =00003D DBD}Given square matrices $A$ and $B$
whose diagonal symbols are distinct from off-diagonal symbol. If there
exists a permutation matrix $P$ such that $PAP^{T}=B$ and pairs
of column vectors $\left(u,v\right)$ and $\left(x,y\right)$ such
that $Pu=v$ and $Px=y$. Then 
\[
P\left(SymMult\left(D_{u},A,D_{x}\right)\right)P^{T}=SymMult\left(D_{v},B,D_{y}\right)
\]
 where $D_{u}=diag(u)$, $D_{v}=diag(v)$, $D_{x}=diag(x)$, and $D_{y}=diag(y)$
are diagonal matrices. Further, 
\[
\Pi\left(SymMult\left(D_{u},A,D_{x}\right)\right)\preceq\Pi\left(A\right)\mbox{ and }\Pi\left(SymMult\left(D_{v},B,D_{y}\right)\right)\preceq\Pi\left(B\right).
\]
\end{thm}

\begin{proof}
Assume $PAP^{T}=B$ and $\left(u,v\right)$ and $\left(x,y\right)$
are column vector pairs such that $Pu=v$ and $Px=y$. Then $PD_{u}P^{T}=D_{v}$
and $PD_{x}P^{T}=D_{y}$ where $D_{u}=diag(u)$, $D_{v}=diag(v)$,
$D_{x}=diag(x)$, and $D_{y}=diag(y)$ are diagonal matrices. 

Using regular matrix multiplication, multiplying $PAP^{T}=B$ on the
left and right yields
\begin{eqnarray}
PAP^{T} & = & B\label{eq:First}\\
\left(PD_{u}P^{T}\right)PAP^{T}\left(PD_{x}P^{T}\right) & = & D_{v}BD_{y}\label{eq:Second}\\
P\left(D_{u}AD_{x}\right)P^{T} & = & D_{v}BD_{y}\label{eq:Third}
\end{eqnarray}
so $P\left(D_{u}AD_{x}\right)P^{T}=D_{v}BD_{y}$. 

From (\ref{eq:Second}), the output arrays of canonical inner product
strings of 
\[
SymMult\left(\left(PD_{u}P^{T}\right),\left(PAP^{T}\right),\left(PD_{x}P^{T}\right)\right)\mbox{ and }SymMult\left(D_{v},B,D_{y}\right)
\]
 are identical. To go from (\ref{eq:Second}) to (\ref{eq:Third})
requires 

\[
P\left(SymMult\left(D_{u},A,D_{x}\right)\right)P^{T}=SymMult\left(\left(PD_{u}P^{T}\right),\left(PAP^{T}\right),\left(PD_{x}P^{T}\right)\right).
\]
 This is true since symmetric permutations move values around without
changing them. So a symmetric permutation can move strings as easily
as real or complex numbers. Therefore, 

\[
P\left(SymMult\left(D_{u},A,D_{x}\right)\right)P^{T}=SymMult\left(D_{v},B,D_{y}\right).
\]

It is easy to see from (\ref{eq:SysMult(DMD)}) that $\Pi\left(SymMult\left(D_{u},A,D_{x}\right)\right)\preceq\Pi\left(A\right)$.
Let $\alpha$ be the symbol associated with one cell of $A$. Locations
within that cell are the only locations that will have $\alpha$ as
the middle symbol on the lhs of (\ref{eq:SysMult(DMD)}). Therefore,
$\Pi\left(SymMult\left(D_{u},A,D_{x}\right)\right)$ is a refinement
of $\Pi\left(A\right)$. A similar argument holds for $\Pi\left(SymMult\left(D_{v},B,D_{y}\right)\right)$
and $\Pi\left(B\right)$. Therefore $\Pi\left(SymMult\left(D_{u},A,D_{x}\right)\right)\preceq\Pi\left(A\right)$
and $\Pi\left(SymMult\left(D_{v},B,D_{y}\right)\right)\preceq\Pi\left(B\right)$
as desired.
\end{proof}
For PCMs it can be shown that the pattern from the first symbolic
squaring is the same as symbolically applying the diagonal from the
left and right.
\begin{thm}
\label{thm:FirstDMD}Given $m^{2}\times m^{2}$ PCM $M=D+R$ where
$D$ is the diagonal and $R=I\otimes\left(1*\left(J-I\right)\right)+\left(J-I\right)\otimes\left(2*I\right)$
are the off-diagonal values. Then 
\[
\Pi\left(SymMult\left(D,M,D\right)\right)=\Pi\left(SymSqr\left(M\right)\right).
\]
\end{thm}

\begin{proof}
The result is by careful consideration of the canonical inner product
strings for each $(i,j)$ location. 

First convert the off-diagonal values of $R$ into symbols. The off-diagonal
values of $R$ are `1', `2', and `0'. The locations associated with
`0' are not explicitly called out in the equation for $R.$ Rewrite
$R$ to include those locations, 
\[
R=I\otimes\left(\mathbf{1}*\left(J-I\right)\right)+\left(J-I\right)\otimes\left(\mathbf{2}*I+\mathbf{0}*\left(J-I\right)\right).
\]
Now substitute symbols in place of `1', `2', and `0' to get 
\[
\overline{R}=I\otimes\left(\sigma_{1}*\left(J-I\right)\right)+\left(J-I\right)\otimes\left(\sigma_{2}*I+\sigma_{3}*\left(J-I\right)\right).
\]
 $\overline{R}$ has numeric zeros on its diagonal because the splitting
$M=D+R$ places the diagonal terms of $M$ in $D$. 

Let $\overline{D}$ be a diagonal matrix where symbols have been substituted
for values on the diagonal. All off-diagonal locations of $\overline{D}$
are a numeric zeros, due to the splitting. The $i^{{\rm th}}$ diagonal
symbols is denoted as $\delta_{i}$, i.e., $\overline{D}_{i,i}=\delta_{i}$.

Let $\overline{M}=\overline{D}+\overline{R}$. Symbolically squaring
$M$ can be written as 
\begin{eqnarray*}
SymSqr\left(M\right) & = & \overline{M}\times\overline{M}\\
 & = & \left(\overline{D}+\overline{R}\right)\times\left(\overline{D}+\overline{R}\right)\\
 & = & \overline{D}\overline{D}+\overline{D}\overline{R}+\overline{R}\overline{D}+\overline{R}\overline{R}
\end{eqnarray*}
where the rhs abuses notation by using regular matrix notation applied
to arrays of symbols. Terms in the canonical inner product strings
involving one or more diagonal symbols come from $\overline{D}\overline{D}+\overline{D}\overline{R}+\overline{R}\overline{D}$.
Terms in the canonical inner product strings involving two off-diagonal
symbols come from $\overline{R}\overline{R}$.

Table \ref{tab:Canonical-inner-product-terms} shows all of the terms
making up the canonical inner product strings from symbolically squaring
$\overline{M}$. The first column in Table \ref{tab:Canonical-inner-product-terms}
shows the symbol associated with an $(i,j)$ location in $\overline{M}$.
The second column shows the term(s) involving a diagonal symbol, coming
from $\overline{D}\overline{D}+\overline{D}\overline{R}+\overline{R}\overline{D}$.
The third columns shows all of the terms involving two off-diagonal
symbols, coming from $\overline{R}\overline{R}$. 

To use Table \ref{tab:Canonical-inner-product-terms} to construct
the canonical inner product string at $(i,j)$, first find the row
matching the symbol at $\overline{M}_{i,j}$ in the first column.
Then using that row, concatenate the term(s) in the second column
followed those in the third column after lexicographically ordering
the terms in the third column. Note that the term(s) in the second
column are already in the order specified for terms involving diagonal
symbols. 
\begin{table}
\begin{tabular}{|c|c|c|}
\hline 
$\overline{M}_{i,j}$ & $\overline{D}\overline{D}+\overline{D}\overline{R}+\overline{R}\overline{D}$ & $\overline{R}\overline{R}$\tabularnewline
\hline 
\hline 
$\delta_{i}$ & $\delta_{i}\delta_{i}$ & $\sigma_{1}\sigma_{1}\left(m-1\right)+\sigma_{2}\sigma_{2}\left(m-1\right)+\sigma_{3}\sigma_{3}\left(m-1\right)^{2}$\tabularnewline
\hline 
$\sigma_{1}$ & $\delta_{i}\sigma_{1}+\sigma_{1}\delta_{j}$ & $\sigma_{1}\sigma_{1}\left(m-2\right)+\sigma_{2}\sigma_{3}\left(m-1\right)+\sigma_{3}\sigma_{2}\left(m-1\right)+\sigma_{3}\sigma_{3}\left(m-1\right)\left(m-2\right)$\tabularnewline
\hline 
$\sigma_{2}$ & $\delta_{i}\sigma_{2}+\sigma_{2}\delta_{j}$ & $\sigma_{1}\sigma_{3}\left(m-1\right)+\sigma_{3}\sigma_{1}\left(m-1\right)+\sigma_{2}\sigma_{2}\left(m-2\right)+\sigma_{3}\sigma_{3}\left(m-1\right)\left(m-2\right)$\tabularnewline
\hline 
$\sigma_{3}$ & $\delta_{i}\sigma_{3}+\sigma_{3}\delta_{j}$ & $\sigma_{1}\sigma_{3}\left(m-1\right)+\sigma_{3}\sigma_{1}\left(m-1\right)+\sigma_{3}\sigma_{2}\left(m-2\right)+\sigma_{2}\sigma_{3}\left(m-2\right)+\sigma_{3}\sigma_{3}\left(m-2\right)^{2}$\tabularnewline
\hline 
\end{tabular}

\caption{\label{tab:Canonical-inner-product-terms}Canonical inner product
string terms from symbolically squaring a PCM.}
\end{table}

Examining the third column of Table \ref{tab:Canonical-inner-product-terms},
one notes that the terms come in matching sets. For example, using
the row with $\sigma_{2}$ in the first column, the third column has
$(m-1)$ copies of $\sigma_{1}\sigma_{3}$ and $(m-1)$ copies of
$\sigma_{3}\sigma_{1}$. In all cases if a row in column three contains
$k$ copies of $\alpha\beta$ then it also contains $k$ copies of
$\beta\alpha$. So the contribution to the canonical inner product
strings at $(i,j)$ and $(j,i)$ from $\left(\overline{R}\overline{R}\right)_{i,j}$
is identical to the contribution from $\left(\overline{R}\overline{R}\right)_{j,i}$
implying $\overline{R}\overline{R}$ is symmetric. Further, one sees
from Table \ref{tab:Canonical-inner-product-terms} that each off-diagonal
symbol of $\overline{R}$ (the $\sigma_{k}$ in column one) is associated
with a distinct multiset of terms in $\overline{R}\overline{R}$ (column
three of Table \ref{tab:Canonical-inner-product-terms}). The diagonal
of $\overline{R}\overline{R}$ is associated with a fourth multiset
of terms that is distinct from the three off-diagonal multisets. So,
$\left(\sigma_{4}I+\overline{R}\right)$ and $\overline{R}\overline{R}$
have identical patterns,
\[
\Pi\left(\sigma_{4}I+\overline{R}\right)=\Pi\left(\overline{R}\overline{R}\right),
\]
 where the symbol $\sigma_{4}$ is chosen to be distinct from all
other symbols in $\overline{M}$.

Now focus on the second column of Table \ref{tab:Canonical-inner-product-terms}.
Assume that $\overline{M}_{i,j}=\sigma_{1}$ and $\overline{M}_{r,s}=\sigma_{1}$
where $(i,j)\neq(r,s)$. Then the only difference in the canonical
inner product strings are the terms involving diagonal symbols, $\delta_{i}\sigma_{1}+\sigma_{1}\delta_{j}$
and $\delta_{r}\sigma_{1}+\sigma_{1}\delta_{s}$ respectively. Further,
note that the row associated with an off-diagonal symbol use the same
off-diagonal symbol in column two. So if column one has off-diagonal
symbol $\sigma_{k}$, $k\in\left\{ 1,2,3\right\} $, then column two
has the terms $\delta_{i}\sigma_{k}+\sigma_{k}\delta_{j}$. Each two
factor pair of terms $\delta_{i}\sigma_{k}+\sigma_{k}\delta_{j}$
could be replaced by a single three factor term $\delta_{i}\sigma_{k}\delta_{j}$
without changing its ability to distinguish canonical inner product
strings. Further, $\delta_{i}\delta_{i}$ can be replaced by $\delta_{i}\sigma_{4}\delta_{i}$
without changing its ability to distinguish canonical inner product
strings. Performing these substitutions allows column two to be expressed
as $\overline{D}\left(\sigma_{4}I+\overline{R}\right)\overline{D}$.
Now comparing the pattern of $\overline{D}\left(\sigma_{4}I+\overline{R}\right)\overline{D}$
to the pattern of $\left(\overline{D}\overline{D}+\overline{D}\overline{R}+\overline{R}\overline{D}\right)$,
one sees that they are equal,
\[
\Pi\left(\overline{D}\left(\sigma_{4}I+\overline{R}\right)\overline{D}\right)=\Pi\left(\overline{D}\overline{D}+\overline{D}\overline{R}+\overline{R}\overline{D}\right).
\]

Note that $\Pi\left(\overline{D}\left(\sigma_{4}I+\overline{R}\right)\overline{D}\right)\preceq\Pi\left(\sigma_{4}I+\overline{R}\right)$
and 
\[
\Pi\left(\overline{D}\left(\sigma_{4}I+\overline{R}\right)\overline{D}\right)=\Pi\left(\overline{D}\left(\overline{D}+\overline{R}\right)\bar{D}\right)=\Pi\left(\overline{D}\overline{M}\overline{D}\right)
\]
since $\Pi\left(\overline{D}\overline{D}\overline{D}\right)=\Pi\left(\overline{D}\left(\sigma_{4}I\right)\overline{D}\right)$.
Therefore, 
\[
\Pi\left(SymMult\left(D,M,D\right)\right)=\Pi\left(SymSqr\left(M\right)\right)
\]
 for PCM $M$ as was to be shown.
\end{proof}
Theorem \ref{thm:FirstDMD} only applies to the first symbolic squaring
of a PCM. Attempting subsequent symbolic squaring by symbolically
multiplying on the left and right by the new diagonal $D^{3}$ does
not refine the pattern since it has the same pattern as $D$, $\Pi\left(D^{3}\right)=\Pi\left(D\right)$.

\subsection{\label{subsec:PCMs-and-Oracles}PCMs and Oracles}

In Algorithm \ref{alg:Blind-p-similarity-algorithm}, a PCM is constructed
for each input matrix, $A$ and $B$. Let $S=PCM(A_{C})$ and $T=PCM(B_{C})$.
Then symbolic squaring is applied to $S$ and $T$ simultaneously
and the canonical inner product strings compared across products.
This section examines how many columns an Oracle can match between
$S^{*}$ and $T^{*}$ when restricted to permutations of the form
$P\otimes P$. 
\begin{thm}
\label{thm:AllOrNone}Given $m^{2}\times m^{2}$ PCMs $S$ and $T$
and their refined patterns $\Pi\left(S^{*}\right)$ and $\Pi\left(T^{*}\right)$
respectively. An Oracle asked to match columns of $S^{*}$and $T^{*}$
using permutations of the form $P\otimes P$ will either match all
of the columns or none of the columns.
\end{thm}

\begin{proof}
Proof by contradiction. Assume $S\underset{P}{\nsim}T$ and the Oracle
matches $k$, $1\leq k<m^{2}$, columns of $S^{*}$ and $T^{*}$ using
a permutation $P_{O}\in\left\{ P\otimes P\right\} $. Then $k$ columns
of $P_{O}S^{*}P_{O}^{T}$ and $T^{*}$ match. Symmetrically permute
the matched columns to the lhs using $P_{lhs}$ to get 
\[
P_{lhs}P_{O}S^{*}P_{O}^{T}P_{lhs}^{T}=\left[\begin{array}{c|c}
C1 & C2^{T}\\
\hline C2 & S2^{*}
\end{array}\right]\mbox{ and }P_{lhs}T^{*}P_{lhs}^{T}=\left[\begin{array}{c|c}
C1 & C2^{T}\\
\hline C2 & T2^{*}
\end{array}\right]
\]
 where 
\[
\left[\begin{array}{c}
C1\\
C2
\end{array}\right]
\]
 are the matched columns. Both $P_{lhs}P_{O}S^{*}P_{O}^{T}P_{lhs}^{T}$
and $P_{lhs}T^{*}P_{lhs}^{T}$ are stable patterns, so symbolic squaring
does not change the pattern.

Now apply the permutations to $S$ and $T$ to get 
\[
S^{(0)}=P_{lhs}P_{O}SP_{O}^{T}P_{lhs}^{T}
\]
 and 
\[
T^{(0)}=P_{lhs}TP_{lhs}^{T}
\]
 which have a conforming block structure 
\[
S^{(0)}=\left[\begin{array}{c|c}
C1^{(0)} & \left(C2^{(0)}\right)^{T}\\
\hline C2^{(0)} & S2^{(0)}
\end{array}\right]\mbox{ and }T^{(0)}=\left[\begin{array}{c|c}
C1^{(0)} & \left(C2^{(0)}\right)^{T}\\
\hline C2^{(0)} & T2^{(0)}
\end{array}\right].
\]
$S^{(0)}$ and $T^{(0)}$ have identical off-diagonal structure. The
application of $P_{O}$ does not change the off-diagonal structure
because $P_{O}\in\left\{ P\otimes P\right\} $ is an automorphism
of $R$. Then $P_{lhs}$ is applied to $P_{O}SP_{O}^{T}$ and $T$,
so the off-diagonal structure changes the same way for both. The diagonal
of $C1^{(0)}$ for $S^{(0)}$ and $T^{(0)}$ must match. If the diagonals
in $C1^{(0)}$ differ, they will also differ in $C1$ of $P_{lhs}P_{O}S^{*}P_{O}^{T}P_{lhs}^{T}$
and $P_{lhs}T^{*}P_{lhs}^{T}$ creating a contradiction.

Now consider the first symbolic squaring. From Theorem \ref{thm:FirstDMD}
the pattern from the first symbolic squaring matches the pattern from
$SymMult\left(D,M,D\right)$. This still holds after symmetrically
permuting columns to the lhs. So $SymSqr\left(S^{(0)}\right)$ and
$SymSqr\left(T^{(0)}\right)$ have patterns matching the patterns
from $SymMult\left(D_{S^{(0)}},S^{(0)},D_{S^{(0)}}\right)$ and $SymMult\left(D_{T^{(0)}},T^{(0)},D_{T^{(0)}}\right)$
respectively. In block form that is 
\[
D_{S^{(0)}}S^{(0)}D_{S^{(0)}}=\left[\begin{array}{c|c}
D_{C1^{(0)}}\\
\hline  & D_{S2^{(0)}}
\end{array}\right]\times\left[\begin{array}{c|c}
C1^{(0)} & \left(C2^{(0)}\right)^{T}\\
\hline C2^{(0)} & S2^{(0)}
\end{array}\right]\times\left[\begin{array}{c|c}
D_{C1^{(0)}}\\
\hline  & D_{S2^{(0)}}
\end{array}\right]
\]
 and 
\[
D_{T^{(0)}}T^{(0)}D_{T^{(0)}}=\left[\begin{array}{c|c}
D_{C1^{(0)}}\\
\hline  & D_{T2^{(0)}}
\end{array}\right]\times\left[\begin{array}{c|c}
C1^{(0)} & \left(C2^{(0)}\right)^{T}\\
\hline C2^{(0)} & T2^{(0)}
\end{array}\right]\times\left[\begin{array}{c|c}
D_{C1^{(0)}}\\
\hline  & D_{T2^{(0)}}
\end{array}\right].
\]

Focus on what happens to $C2^{(0)}$. We see that the patterns resulting
from symbolic multiplication of $D_{S2^{(0)}}C2^{(0)}D_{C1^{(0)}}$
and $D_{T2^{(0)}}C2^{(0)}D_{C1^{(0)}}$ do not match since $D_{S2^{(0)}}\neq D_{T2^{(0)}}$.
Recall that $S$ and $T$ are not p-similar by hypothesis. Therefore
the $C2$ from $P_{lhs}P_{O}S^{*}P_{O}^{T}P_{lhs}^{T}$ and $P_{lhs}T^{*}P_{lhs}^{T}$
cannot match, implying the Oracle failed to match any columns of $S^{*}$
and $T^{*}$ using a permutation of the form $P\otimes P$.

So the Oracle either matches all columns of $S^{*}$ and $T^{*}$
or none of the columns using amenable permutations.
\end{proof}
There is no partial matching of $S^{*}$ and $T^{*}$ using a permutation
of the form $P\otimes P$. Either $S^{*}$ and $T^{*}$ are p-similar
and the columns can be matched or $S^{*}$ and $T^{*}$ are not p-similar
and no columns can be matched using a permutation of the form $P\otimes P$.
This is independent of comparing multisets of diagonal symbols for
$S^{*}$ and $T^{*}$.

\subsection{Symbolic Squaring of $\mathbf{PCM\left(\left(A\oplus B\right)_{C}\right)}$}

Theorem \ref{thm:AllOrNone} showed that given two matrices $A$ and
$B$ and their PCMs $S$ and $T$, where $S=PCM\left(A_{C}\right)$
and $T=PCM\left(B_{C}\right)$, that either $S^{*}$ is p-similar
to $T^{*}$ or none of the columns can be matched using an admissible
permutation of the form $P\otimes P$. However, it does not say anything
about diagonal symbols.

Given non p-similar $m\times m$ matrices $A$ and $B$, construct
the PCM of the direct sum, $M=PCM\left(\left(A\oplus B\right)_{C}\right)$.
Also construct the PCMs $S=PCM\left(A_{C}\right)$ and $T=PCM\left(B_{C}\right)$.
The color matrix of the direct sum $\left(A\oplus B\right)$ is given
by 
\[
\left(A\oplus B\right)_{C}=\left[\begin{array}{cc}
A_{C} & 3J\\
3J & B_{C}
\end{array}\right].
\]
 Applying column major ordering to construct the diagonal of $M$
results in $2m\times2m$ diagonal blocks where the $k^{{\rm th}}$
diagonal block in the upper left quadrant of $M$ is given by 
\begin{equation}
\left[\begin{array}{c|c}
\left(diag\left(A_{C}(:,k)\right)+(J-I)\right) & J\\
\hline J & \left(J+2I\right)
\end{array}\right]_{2m\times2m}\label{eq:ULS}
\end{equation}
 and the $\left(m+k\right)^{{\rm th}}$ diagonal block in the lower
right quadrant of $M$ is given by 
\begin{equation}
\left[\begin{array}{c|c}
\left(J+2I\right) & J\\
\hline J & \left(diag\left(B_{C}(:,k)\right)+(J-I)\right)
\end{array}\right]_{2m\times2m}\label{eq:LRT}
\end{equation}
 where $J$ is the $m\times m$ matrix of all ones and $I$ is the
conforming identity matrix. Off-diagonal blocks of $M$ are given
by $2I_{2m\times2m}$. 

So, all of the entries in $S$ and $T$ are embedded in $M$. The
$k^{{\rm th}}$ $m\times m$ diagonal block of $S$, $\left(diag\left(A_{C}(:,k)\right)+(J-I)\right)$,
is the upper left portion of each $2m\times2m$ diagonal block in
the upper left quadrant of $M$, see (\ref{eq:ULS}). The $m\times m$
off-diagonal blocks of $S$ are the upper left portion of each $2m\times2m$
off-diagonal blocks in the upper left quadrant of $M$. A similar
case holds for $T$ except the blocks occupy the lower right portion
of each $2m\times2m$ block in the lower right quadrant of $M$.

Perform symbolic squaring on $M$ until the stable pattern $\Pi\left(M^{*}\right)$
is reached. By the argument in Section \ref{subsec:ConjugateDirectionDescent},
$\Pi\left(M^{*}\right)=\Pi\left(M_{Aut}\right)$ so each cell of $M^{*}$
represents an individual orbit. The patterns associated with locations
of $S$ and $T$ in $M^{*}$ are identical to $\Pi\left(S^{*}\right)$
and $\Pi\left(T^{*}\right)$ since $\Pi\left(S^{*}\right)$ and $\Pi\left(T^{*}\right)$
are the orbits of $S$ and $T$ respectively. Therefore, Theorem \ref{thm:AllOrNone}
applies to the columns of $M^{*}$ associated with $S$ and $T$. 

Since $A$ and $B$ are non p-similar, none of the columns in $M^{*}$
associated with $S$ and $T$ can be matched. Therefore the multisets
of diagonal symbols corresponding to locations associated with $S$
and $T$ in $M^{*}$ are distinct. Then since the patterns in $M^{*}$
associated with $S$ and $T$ are identical to $\Pi\left(S^{*}\right)$
and $\Pi\left(T^{*}\right)$, the diagonal mixes from $S^{*}$ and
$T^{*}$ are different as long as consistent symbol substitution is
performed each symbolic squaring. 

This is the final piece showing that comparing the multisets of diagonal
symbols from $S^{*}$ and $T^{*}$ is necessary and sufficient to
determine whether $A$ and $B$ are p-similar.

\section{\label{sec:FindingP}Finding $\mathbf{P}$ using the Blind Algorithm}

Given that Algorithm \ref{alg:Blind-p-similarity-algorithm} is a
polynomial blind permutation similarity algorithm, it can be used
as part of a second polynomial algorithm to find a permutation between
p-similar matrices. Such an algorithm is presented in Algorithm \ref{alg:FindingP}.
The remainder of this section describes the concepts behind Algorithm
\ref{alg:FindingP}.
\begin{algorithm}
\begin{lstlisting}
01 function [psim, p] =  FindPUsingBPSAY(M1, M2) 
02 % Find P Using BPSAY 
03 % Inputs: 
04 %   M1 - square (real or complex) matrix 
05 %   M2 - square (real or complex) matrix 
06 % Outputs: 
07 %   psim - boolean (TRUE if M1 & M2 are p-similar) 
08 %   p - permutation vector such that M1(p,p) == M2 
09   
10 m = size(M1, 1); 
11 p = [1:m]'; 				% initialize permutation vector
12  
13 psim = BPSAY(M1, M2); 		% check p-similarity
14 if (~psim), return(psim, p); end if  % not p-similar 
15  
16 [A, B] = SymbolSubstitution(M1, M2); % positive integer matrices
17 A = ShiftAndTranslate(A, beta=m^2, gamma=0); % make diag symbols distinct
18 B = ShiftAndTranslate(B, beta=m^2, gamma=0); % make diag symbols distinct
19  
20 % ASSERT: A and B are p-similar w/diag symbols distinct from off-diag symbols
21  
22 n = m; 		% size of trailing principal submatrix
23 for c = 1:(m-1) 	% search trailing principal submatrix for p(c)
24    DBD = SymMult(diag(B(:,1)), B, diag(B(1,:)));      
25    for j = 1:n 
26        if (A(j,j) ~= B(1,1)), continue; end if
27         
28        AJ1 = ExchangeIJ(A, 1, j); 	% exchange rows/cols 1 <-> j
29        DAJ1D = SymMult(diag(AJ1(:,1)), AJ1, diag(AJ1(1,:))); 
30        [S, T] = SymbolSubstitution(DAJ1D, DBD); 
31
32        % trailing principal submatrices
33        S22 = S(2:n,2:n);  T22 = T(2:n,2:n);  
34         
35        psim = BPSAY(S22, T22); % check p-similarity
36        if (psim) 		  % column p(c+j-1) is amenable 
37            p([c (c+j-1)]) = p([(c+j-1) c]); % p(c) <-> p(c+j-1) 
38            A = S22;  B = T22;  % new A and B are p-similar
39            break; 		  % move to next column
40        end if
41    end 
42    n = n - 1;      % update size of trailing principal submatrices 
43 end
44 
45 return(psim, p);
46 end
\end{lstlisting}

\caption{\label{alg:FindingP}Pseudocode for finding the permutation using
BPSAY, Algorithm \ref{alg:Blind-p-similarity-algorithm}}

\end{algorithm}

\subsection{\label{subsec:Overall-Concept}Overall Concept}

Algorithm \ref{alg:FindingP} uses the blind permutation similarity
algorithm, Algorithm \ref{alg:Blind-p-similarity-algorithm}, to find
a permutation vector $p$, such that $M1(p,p)=M2$ when $M1$ and
$M2$ are p-similar. If $M1$ and $M2$ are not p-similar, then $psim=FALSE$
and $p=\left[1:m\right]^{T}$, the identity permutation vector, are
returned as output. 

The main loop in Algorithm \ref{alg:FindingP}, lines 23-43, looks
for the column to use as $p(c)$, where $c$ is the iteration variable.
The preconditions for each iteration are that i) $A$ and $B$ are
p-similar, ii) their diagonal symbols are distinct from the off-diagonal
symbols, and iii) the first $\left(c-1\right)$ locations of $p$,
$p\left(1:\left(c-1\right)\right)$, are correct. 

Upon entering the loop for the first time with $c=1$, the preconditions
are satisfied since:
\begin{enumerate}
\item $A$ and $B$ are p-similar by the call to Algorithm \ref{alg:Blind-p-similarity-algorithm}
on line 13 of Algorithm \ref{alg:FindingP}. 
\item The transformations applied in lines 16-18 are similar to those constructing
color matrices from the input matrices. So the diagonal symbols of
$A$ and $B$ are distinct from off-diagonal symbols. 
\item $\left(c-1\right)=0$ columns have been matched correctly, satisfying
the third precondition. 
\end{enumerate}
For each iteration, $A$ and $B$ are the remaining trailing principal
submatrices to be matched. Only columns of $A$ with a diagonal symbol
$A_{j,j}=B_{1,1}$ are potential matches for the first column of $B$.
The nested loop, lines 25-41 searches down the diagonal of $A$ for
columns where $A_{j,j}=B_{1,1}$. When it finds such a column, it
is symmetrically permuted to become the first column of $AJ1$. In
Algorithm \ref{alg:FindingP} the exchange happens on line 28. The
exchange can be expressed as 
\[
AJ1=P_{1j}AP_{1j}^{T}
\]
 where permutation matrix $P_{1j}$ exchanges rows and columns $1$
and $j$. Assuming the first column of $AJ1$ is a correct match for
the first column of $B$. Then there exists a complementary permutation
$P_{c}$ such that 
\[
P_{c}\left(AJ1\right)P_{c}^{T}=B
\]
 where the structure of $P_{c}$ is given by 
\[
P_{c}=\left[\begin{array}{c|c}
1\\
\hline  & P22
\end{array}\right].
\]
 Then 
\begin{equation}
\left[\begin{array}{c|c}
1\\
\hline  & P22
\end{array}\right]\times\left[\begin{array}{c|c}
AJ1_{1,1} & AJ1_{1,2:n}\\
\hline AJ1_{2:n,1} & AJ1_{2:n,2:n}
\end{array}\right]\times\left[\begin{array}{c|c}
1\\
\hline  & P22
\end{array}\right]^{T}=\left[\begin{array}{c|c}
B_{1,1} & B_{1,2:n}\\
\hline B_{2:n,1} & B_{2:n,2:n}
\end{array}\right]\label{eq:P22}
\end{equation}
 where 
\begin{eqnarray}
P22\times AJ1_{2:n,1} & = & B_{2:n,1}\label{eq:ColConstraint}\\
AJ1_{1,2:n}\times P22^{T} & = & B_{1,2:n}\label{eq:RowConstraint}
\end{eqnarray}
 and 
\begin{eqnarray}
P22\times AJ1_{2:n,2:n}\times P22^{T} & = & B_{2:n,2:n}.\label{eq:PAJIP'=00003DB}
\end{eqnarray}

Although $P22$ is unknown, Algorithm \ref{alg:Blind-p-similarity-algorithm}
can be used to determine if $AJ1_{2:n,2:n}$ and $B_{2:n,2:n}$ in
(\ref{eq:PAJIP'=00003DB}) are p-similar. However, just testing the
$\left(2:n,2:n\right)$ trailing principal submatrices for p-similarity
may give a false result. It may be that the only permutation(s), that
match the trailing principal submatrices, violate (\ref{eq:ColConstraint})
and (\ref{eq:RowConstraint}). To eliminate this possibility, the
first column and first row of $AJ1$ and $B$ are used to refine the
$\left(2:n,2:n\right)$ trailing principal submatrices prior to checking
if they are p-similar. For $AJ1$, this is done by constructing diagonal
matrices from the first column and first row of $AJ1$, $diag\left(AJ1_{:,1}\right)$
and $diag\left(AJ1_{1,:}\right)$ respectively, and symbolically applying
them from the left and the right to $AJ1$, line 29. The $\left(2:n,2:n\right)$
trailing principal submatrix of $B$ is similarly refined using the
first column and first row of $B$, line 24. Theorem \ref{thm:P(DAD)P' =00003D DBD}
guarantees that the symbolic multiplication will refine the $\left(2:n,2:n\right)$
trailing principal submatrices while preserving the p-similarity relationship
between $AJ1$ and $B$. The left and right diagonal matrices are
constructed from the first column an first row respectively because
$A$ and $B$ might not be symmetric. 

Algorithm \ref{alg:Blind-p-similarity-algorithm} is applied to $S22$
and $T22$, the refined $(2:n,2:n)$ trailing principal submatrices
of $DAJ1D$ and $DBD$ respectively, to test for p-similarity without
the possibility of a false response, line 35. If $S22$ and $T22$
are p-similar, then the amenable permutation also satisfies (\ref{eq:ColConstraint})
and (\ref{eq:RowConstraint}). 

After finding a $S22$ and $T22$ that are p-similar, the global index
of the column permuted into the first column of $AJ1$ is saved as
$p(c)$, line 37, and $A$ and $B$ are set to $S22$ and $T22$ respectively,
line 38. The size of the trailing principal submatrix is decremented
in preparation for searching for the next column, line 42. 

P-similar $S22$ and $T22$ are guaranteed to exist because the precondition
for entering the loop is that $A$ and $B$ are p-similar. 

For the next iteration $c=2$, again the preconditions are met. $A$
and $B$ are p-similar, tested on line 35, their diagonal symbols
are distinct from off-diagonal symbols since the inputs to the symbol
substitution, line 30, have distinct diagonal symbols, and the first
$\left(c-1\right)$ locations in $p$, p(1), are correct. 

Therefore, for p-similar input matrices $M1$ and $M2$, Algorithm
\ref{alg:FindingP} finds a permutation vector $p$ such that $M1(p,p)=M2$
in $O(m^{2})$ iterations.

\subsection{Algorithm \ref{alg:FindingP} Complexity}

Given $m\times m$ input matrices, the nested loops in Algorithm \ref{alg:FindingP}
make $O\left(m^{2}\right)$ calls to BPSAY, Algorithm \ref{alg:Blind-p-similarity-algorithm}.
So the overall complexity of Algorithm \ref{alg:FindingP} is $O\left(m^{14}\right)$
using symbolic squaring. If widely space primes matrices are used,
Section \ref{subsec:MatMultIsMonotonic}, the overall complexity drops
to $O\left(m^{12}\right)$. Therefore finding the permutation between
two p-similar matrices is polynomial.

\section{\label{sec:Discussion}Discussion}

In practice, it does not appear to be necessary to add edge weights
to the PCG. The edge weights are added to eliminate the possibility
of exchanging the rows with the columns. Edge weights simplified the
arguments by avoiding the need to carry the transpose based permutations
through the proofs. If rook's graphs are compact then color refinement
works on PCMs with equal edge weights \cite{tinhofer86,brualdi88,tinhofer91}.
Although rook's graphs have many of the characteristics of compact
graphs, the author has not been able to show that rook's graphs are
compact.

\subsection{\label{subsec:Origin-of-PCGs}Origin of PCGs}

For a period of time the author tried to develop a graph isomorphism
algorithm. For every algorithm, non isomorphic graphs were found that
were not discriminated. Eventually the investigation turned to looking
at why color refinement fails.

Given two non-isomorphic graphs, iterative color refinement can be
applied to the adjacency matrices until they converged to a stable
partition \cite{arvind2017}. It is well known that strongly regular
graphs cannot be separated using naive color refinement. However,
one could ask an Oracle to maximally column match the matrices. The
result would look like 
\[
A=\left[\begin{array}{c||c}
C1 & C2^{T}\\
\hline\hline C2 & A2
\end{array}\right]\mbox{ and }B=\left[\begin{array}{c||c}
C1 & C2^{T}\\
\hline\hline C2 & B2
\end{array}\right]
\]
 after symmetrically permuting the matched columns to the lhs. Differences
are confined to the lower right blocks, where $A2\neq B2$. We can
assume the column multisets match (otherwise $A$ and $B$ are known
to be non-isomorphic). Next consider ordering $C2$ into groups of
identical rows. The permutation to do this can be applied symmetrically
to both $A$ and $B$. Now the matrices look something like 
\begin{equation}
A=\left[\begin{array}{c||c|c|c|c}
C1 & r_{1}^{T} & r_{2}^{T} & r_{3}^{T} & r_{4}^{T}\\
\hline\hline r_{1} & \delta & \gamma & \alpha & \beta\\
\hline r_{2} & \gamma & \delta & \beta & \alpha\\
\hline r_{3} & \alpha & \beta & \delta & \gamma\\
\hline r_{4} & \beta & \alpha & \gamma & \delta
\end{array}\right]\mbox{ and }B=\left[\begin{array}{c||c|c|c|c}
C1 & r_{1}^{T} & r_{2}^{T} & r_{3}^{T} & r_{4}^{T}\\
\hline\hline r_{1} & \delta & \gamma & \beta & \alpha\\
\hline r_{2} & \gamma & \delta & \alpha & \beta\\
\hline r_{3} & \beta & \alpha & \delta & \gamma\\
\hline r_{4} & \alpha & \beta & \gamma & \delta
\end{array}\right]\label{eq:differences}
\end{equation}
 where $\alpha\neq\beta$ and each $r_{j}$ is a pairwise distinct
row vector. Consider the first column of the unmatched region. Notice
that $A$ has $\alpha$ at the intersection of $r_{3}$ and $r_{1}^{T}$
whereas $B$ has a $\beta$ at that location and the reverse at the
intersection of $r_{4}$ and $r_{1}^{T}$. By hypothesis, $r_{1}\neq r_{3}$,
$r_{3}\neq r_{4}$, and $r_{4}\neq r_{1}$, so the author looked for
ways to take advantage of this situation by constructing invariants
such as all of the four-cycles. Four-cycles are permutation independent.
The process always broke down when trying to compare results using
multisets. The multisets either maintained row constraints and relaxed
column constraints or maintained column constraints and relaxed row
constraints. This happens for any comparison method based on multisets
such as those described in \cite{mckay2014,arvind2017}.

This led to asking the question: ``Can symmetric permutation constraints
be added to the original graphs?'' If so, the symmetric permutation
constraints would be part of the original problem and using multisets
for comparison should not experience the same difficulty. PCGs were
the result. PCGs also changed the focus from graph isomorphism to
permutation similarity.

\subsection{\label{subsec:MatMultIsMonotonic}Refinement using Regular Matrix
Multiplication}

It is not clear that the sequence of refinements generated by symbolic
squaring can be replicated using regular matrix multiplication. If
the sequence cannot be replicated, it makes arguments using the existing
machinery for SPD matrices more difficult to accept. It turns out
there are two different ways the sequence of refinements can be replicated
using regular matrix multiplication. They both use matrices constructed
from widely-spaced prime numbers. 
\begin{defn}
\label{def:widely-spaced-primes}A \emph{widely-spaced primes matrix
$W$, where} $W=wspm\left(M\right)$ is constructed from a symmetric
matrix $M$, whose diagonal symbols are distinct from off-diagonal
symbols. The construction uses the following method:
\begin{enumerate}
\item Let $M=\left(\Pi,\,\Sigma,\,g\right)$ be a $m\times m$ symmetric
matrix whose diagonal symbols are distinct from the off-diagonal symbols.
\item Assume $M$ has $n=\left|\Sigma\right|$ distinct symbols where there
are $n_{1}$ off-diagonal symbols and $n_{2}$ diagonal symbols with
$n_{1}+n_{2}=n$.
\item Let $\Pi\left(W\right)=\Pi\left(M\right)$.
\item Choose $n$ prime numbers, $p_{i}$, to be the symbols of $W$ using
the following recurrence.
\begin{enumerate}
\item Let $p_{1}>mk^{2}$ be the first prime number, where $k\in\mathbb{N^{+}}$
(typically $k=1$).
\item For $i=2,\cdots,n$, chose $p_{i}>m*p_{i-1}^{2}$.
\end{enumerate}
\item Use the first $n_{1}$ primes, $p_{1},\ldots,p_{n_{1}}$, as off-diagonal
symbols of $W$ and use the remaining $n_{2}$ primes, $p_{n_{1}+1},\ldots,p_{n}$,
as diagonal symbols of $W$. 
\end{enumerate}
\end{defn}

\begin{rem}
A widely-spaced primes matrix $W$ is SPD since it is symmetric and
strictly diagonally dominant by construction. 
\end{rem}

Consider squaring a $m\times m$ widely-spaced primes matrix $W$.
Let $W=D+O$ where $O$ is the matrix of off-diagonal symbols and
$D$ is the diagonal. Recall that widely-spaced primes matrix has
diagonal symbols that are distinct from off-diagonal symbols. Using
regular matrix multiplication, the value at $(i,j)$ is given by 
\begin{equation}
\left[W\times W\right]_{i,j}=\begin{cases}
D_{i,i}O_{i,j}+O_{i,j}D_{j,j}+\sum_{\begin{array}{c}
k\neq i\\
k\neq j
\end{array}}O_{i,k}O_{k,j} & \mbox{if \ensuremath{i\neq j}}\\
D_{i,i}D_{i,i}+\sum_{\begin{array}{c}
k\neq i\end{array}}O_{i,k}O_{k,i} & \mbox{\mbox{if \ensuremath{i=j}}}
\end{cases}\label{eq:WWij}
\end{equation}
 where terms involving diagonal symbols come first, followed by terms
with two off-diagonal symbols. This is similar to the order used for
canonical inner product strings. 
\begin{rem}
Each term in the inner product $\left[W\times W\right]_{i,j}$ is
the product of two primes.
\end{rem}

~
\begin{rem}
\label{rem:DistinctTerms}Let $\Sigma=\left\{ p_{1},\cdots,p_{n}\right\} $
be the set of widely-spaced primes in $W$. Then the set of all possible
distinct terms, 
\[
\left\{ p_{1}p_{1},\,p_{1}p_{2},\,p_{2}p_{2,}\,p_{1}p_{3},\,p_{2}p_{3},\,p_{3}p_{3},\cdots,p_{1}p_{n},\cdots,p_{n}p_{n}\right\} 
\]
is totally ordered by less than 
\[
p_{1}p_{1}<p_{1}p_{2}<p_{2}p_{2}<p_{1}p_{3}<p_{2}p_{3}<p_{3}p_{3}<\cdots<p_{1}p_{n}<\cdots<p_{n}p_{n}
\]
and the ratio between two adjacent terms is greater than $mp_{1}$,
\[
\frac{p_{i+1}p_{j}}{p_{i}p_{j}}>mp_{i}\mbox{ for }j>i\mbox{ and }\frac{p_{1}p_{j+1}}{p_{j}p_{j}}>mp_{1}\mbox{ for }i=j.
\]
\end{rem}

\begin{lem}
\label{lem:WinnerProdUnique}The inner product $\left[W\times W\right]_{i,j}$
is unique to a set of distinct terms and their multiplicities.
\end{lem}

\begin{proof}
Let $W$ be a $m\times m$ widely-spaced primes matrix. Then each
inner product $\left[W\times W\right]_{i,j}$ has $m$ terms and can
be written as 
\[
\left[W\times W\right]_{i,j}=\sum m_{k}t_{k}
\]
where each $t_{k}$ is a distinct term in the inner product with multiplicity
$m_{k}$ and $\sum m_{k}=m$. 

Proof is by contradiction. Assume $\left[W\times W\right]_{i,j}=\left[W\times W\right]_{r,s}$
where 
\[
\left[W\times W\right]_{r,s}=\sum\overline{m}_{l}\overline{t}_{l}
\]
 and $\sum m_{k}=\sum\overline{m}_{l}=m$ but the sets of distinct
terms $\left\{ t_{k}\right\} $ and $\left\{ \overline{t}_{l}\right\} $
do not match, $\left\{ t_{k}\right\} \neq\left\{ \overline{t}_{l}\right\} $.
Subtract terms appearing in both $\left[W\times W\right]_{i,j}$ and
$\left[W\times W\right]_{r,s}$ from $\sum m_{k}t_{k}$ and $\sum\overline{m}_{l}\overline{t}_{l}$.
Let the resultant sums be $\sigma_{ij}$ and $\overline{\sigma}_{rs}$
respectively. Now $\sigma_{ij}=\overline{\sigma}_{rs}>0$ but they
have no common terms. WLOG assume that $\sigma_{ij}$ contains the
largest distinct term between $\sigma_{ij}$ and $\overline{\sigma}_{rs}$.
Let $p_{k1}p_{k2}$ be that largest term in $\sigma_{ij}$ and $p_{l1}p_{l2}$
be the largest term in $\overline{\sigma}_{rs}$. So $p_{k1}p_{k2}>p_{l1}p_{l2}$.
All of the terms in $\overline{\sigma}_{rs}$ are less than or equal
to $p_{l1}p_{l2}$. Therefore the largest $\overline{\sigma}_{rs}$
can be is $mp_{l1}p_{l2}$. But by construction, the next largest
distinct term greater than $p_{l1}p_{l2}$ is strictly greater than
$mp_{l1}p_{l2}$, see Remark \ref{rem:DistinctTerms}. Therefore $\sigma_{ij}$
cannot equal $\overline{\sigma}_{rs}$ contradicting the assumption
that two inner products of $W\times W$ that are equal can be constructed
from different sets of distinct terms with different multiplicities. 
\end{proof}
A consequence of Lemma \ref{lem:WinnerProdUnique} is that the inner
product $\left[W\times W\right]_{i,j}$ is decomposable using modulo
arithmetic by starting with the largest prime and working downward,
recursively applying modulo arithmetic to both the integer multiple
of the prime and the remainder. 
\begin{lem}
The multiplicity $m_{k}$ of a distinct term, $t_{k}=p_{l}p_{r}=p_{r}p_{l}$,
in the inner product $\left[W\times W\right]_{i,j}$ is equal to the
sum of the multiplicities $n_{1}$ of $\left(p_{l},p_{r}\right)$
and $n_{2}$ of $\left(p_{l},p_{r}\right)$ in the canonical inner
product string $SymSqr\left(W\right)_{i,j}$, $m_{k}=n_{1}+n_{2}$.
\end{lem}

\begin{proof}
Consequence of Lemma \ref{lem:WinnerProdUnique}. The inner product
is uniquely determined by the set of distinct terms in the inner product
and their multiplicities. A distinct term, $t_{k}$, is the product
of a unique pair of primes, $t_{k}=p_{l}p_{r}=p_{r}p_{l}$. Canonical
inner product strings distinguishes between the ordered pairs $\left(p_{l},p_{r}\right)$
and $\left(p_{l},p_{r}\right)$ whereas the inner product does not.
The result follows.
\end{proof}
The next theorem shows that the square of a widely-spaced primes matrix
is a refinement of the widely-spaced primes matrix.
\begin{thm}
\label{thm:SPDMonotonic-1}Given a real $m\times m$ symmetric matrix
$M$ whose diagonal symbols are distinct from off-diagonal symbols,
there exists a widely-spaced primes matrix $W=wspm\left(M\right)$
such that $\Pi\left(W\right)=\Pi\left(M\right)$ and $\Pi\left(W\times W\right)\preceq\Pi\left(W\right)$.
\end{thm}

\begin{proof}
Let $M$ be a real $m\times m$ symmetric matrix whose diagonal symbols
are distinct from off-diagonal symbols. Let $W=wspm\left(M\right)$
be a widely-spaced primes matrix. Then $\Pi\left(W\right)=\Pi\left(M\right)$
by Definition \ref{def:widely-spaced-primes}. 

Let $W=D+O$ where $O$ is the matrix of off-diagonal symbols and
$D$ is the diagonal. To see that $\Pi\left(W\times W\right)\preceq\Pi\left(W\right)$,
it is enough to show that entries that differ in $W$ also differ
in $W\times W$. Consider two off-diagonal locations $(i,j)$ and
$(r,s)$ such that $W_{i,j}\neq W_{r,s}$, i.e., $O_{i,j}\neq O_{r,s}$.
Without loss of generality, assume $i<j$ and $r<s$. Then using the
first row of (\ref{eq:WWij}) 
\begin{eqnarray}
\left[W\times W\right]_{i,j} & = & D_{i,i}O_{i,j}+O_{i,j}D_{j,j}+\sum_{\begin{array}{c}
k\neq i\\
k\neq j
\end{array}}O_{i,k}O_{k,j}\label{eq:DOODij}
\end{eqnarray}
 and 
\begin{eqnarray}
\left[W\times W\right]_{r,s} & = & D_{r,r}O_{r,s}+O_{r,s}D_{s,s}+\sum_{\begin{array}{c}
k\neq r\\
k\neq s
\end{array}}O_{r,k}O_{k,s}\label{eq:DOODrs}
\end{eqnarray}
 where the terms involving two off-diagonal symbols sum to less than
the smallest diagonal symbol $p_{n_{1}+1}$. So if 
\begin{equation}
O_{i,j}\left(D_{i,i}+D_{j,j}\right)\mbox{ is distinct from }O_{r,s}\left(D_{r,r}+D_{s,s}\right)\label{eq:O(D+D)}
\end{equation}
 then $\left[W\times W\right]_{i,j}$ is guaranteed to be distinct
from $\left[W\times W\right]_{r,s}$. To see that the two sides of
(\ref{eq:O(D+D)}) are distinct when $O_{i,j}\neq O_{r,s}$, consider
the following cases:
\begin{casenv}
\item $\left(D_{i,i}+D_{j,j}\right)=\left(D_{r,r}+D_{s,s}\right)$: Then
$O_{i,j}\left(D_{i,i}+D_{j,j}\right)\neq O_{r,s}\left(D_{r,r}+D_{s,s}\right)$
since $O_{i,j}\neq O_{r,s}$ .
\item $\left(D_{i,i}+D_{j,j}\right)\neq\left(D_{r,r}+D_{s,s}\right)$: Without
loss of generality, assume that $\left(D_{i,i}+D_{j,j}\right)<\left(D_{r,r}+D_{s,s}\right)$.
Further, assume that $D_{s,s}$ is maximal among the diagonal symbols
under consideration, $D_{s,s}=max\left(D_{i,i},D_{j,j},D_{r,r},D_{s,s}\right)$,
and that $D_{s,s}=p_{l}$, the $l^{{\rm th}}$ prime number. 
\begin{casenv}
\item ($D_{i,i},\,D_{j,j}<D_{s,s}$) and ($D_{r,r}\leq D_{s,s}$): Substitute
into (\ref{eq:O(D+D)}) to get $O_{i,j}\left(D_{i,i}+D_{j,j}\right)\leq2p_{n_{1}}p_{l-1}<D_{s,s}$
and $O_{r,s}\left(D_{r,r}+D_{s,s}\right)>D_{s,s}$. Therefore $O_{i,j}\left(D_{i,i}+D_{j,j}\right)\neq O_{r,s}\left(D_{r,r}+D_{s,s}\right)$.
\item ($D_{i,i}<D_{r,r}$ \& $D_{j,j}=D_{s,s}$) and ($D_{r,r}<D_{s,s}$):
Substitute into (\ref{eq:O(D+D)}) and rearrange to get $O_{ij}D_{i,i}-O_{r,s}D_{r,r}$
versus $\left(O_{r,s}-O_{i,j}\right)D_{s,s}$. Now $\left|O_{ij}D_{i,i}-O_{r,s}D_{r,r}\right|<2p_{n_{1}}p_{l-1}<D_{s,s}$
and $\left|O_{r,s}-O_{i,j}\right|D_{s,s}>D_{s,s}$. Therefore $O_{i,j}\left(D_{i,i}+D_{j,j}\right)\neq O_{r,s}\left(D_{r,r}+D_{s,s}\right)$.
\item ($D_{i,i}<D_{r,r}$ \& $D_{j,j}=D_{s,s}$) and ($D_{r,r}=D_{s,s}$):
Substitute into (\ref{eq:O(D+D)}) to get $O_{i,j}D_{i,i}$ versus
$\left(2O_{r,s}-O_{i,j}\right)D_{s,s}$. Similar to the prior case,
$O_{i,j}D_{i,i}<2p_{n_{1}}p_{l-1}<D_{s,s}$ and $\left|2O_{r,s}-O_{i,j}\right|D_{s,s}>D_{s,s}$.
Therefore $O_{i,j}\left(D_{i,i}+D_{j,j}\right)\neq O_{r,s}\left(D_{r,r}+D_{s,s}\right)$.
\end{casenv}
\end{casenv}
A similar argument holds when diagonal symbols $W_{i,i}$ and $W_{r,r}$
differ. In this case, only the $D_{i,i}^{2}$ and $D_{r,r}^{2}$ terms
need to be considered since the remaining terms sum to less than the
smallest diagonal symbol $p_{n_{1}+1}$. Inner products for diagonal
and off-diagonal locations differ in $W\times W$. Off-diagonal locations
have two terms involving a diagonal symbol and an off-diagonal symbol
whereas diagonal locations only have one term involving the square
of a diagonal symbol, as shown in (\ref{eq:WWij}). Again, the separation
between primes guarantees they cannot be equal. Therefore if two locations
of $W$ differ, those locations will also differ in $W\times W$ implying
$\Pi\left(W\times W\right)\preceq\Pi\left(W\right)$ as desired.
\end{proof}
Theorem \ref{thm:SPDMonotonic-1} does not say that $\Pi\left(W\times W\right)=\Pi\left(SymSqr\left(M\right)\right)$.
There may be some inner products that symbolic squaring distinguishes
that regular matrix multiplication does not. However, if $SymSqr\left(W\right)$
refines the diagonal, then $W\times W$ also refines the diagonal.
\begin{thm}
\label{thm:WrefinesDiagaonl}For a widely-space primes matrix $W$,
if $SymSqr\left(W\right)$ strictly refines the diagonal, 
\[
\Pi\left(diag\left(SymSqr\left(W\right)\right)\right)\prec\Pi\left(diag\left(W\right)\right),
\]
 then $W\times W$ also strictly refines the diagonal, 
\[
\Pi\left(diag\left(W\times W\right)\right)\prec\Pi\left(diag\left(W\right)\right)
\]
 and $\Pi\left(diag\left(W\times W\right)\right)=\Pi\left(diag\left(SymSqr\left(W\right)\right)\right)$.
\end{thm}

\begin{proof}
The inner product for a diagonal symbol has the same symbol for both
factors of a term, see the second row of (\ref{eq:WWij}). Therefore,
the factor order doesn't matter for the canonical inner product strings
of diagonal symbols and Lemma \ref{lem:WinnerProdUnique} says the
inner product is unique to the set of distinct terms and their multiplicities.
So, if $SymSqr\left(W\right)$ refines the diagonal, $\Pi\left(diag\left(SymSqr\left(W\right)\right)\right)\prec\Pi\left(diag\left(W\right)\right)$,
then $W\times W$ refines the diagonal, $\Pi\left(diag\left(W\times W\right)\right)\prec\Pi\left(diag\left(W\right)\right)$
and 
\[
\Pi\left(diag\left(W\times W\right)\right)=\Pi\left(diag\left(SymSqr\left(W\right)\right)\right)
\]
.
\end{proof}
For off-diagonal locations, there may be cases where $SymSqr\left(W\right)$
refines an off-diagonal location but $W\times W$ does not. However,
if a widely-spaced primes matrix $W$ is generated from PCM $M$,
$M=PCM\left(A_{C}\right)$ and $W=wspm\left(M\right)$, then for the
first squaring, $\Pi\left(W\times W\right)=\Pi\left(SymSqr\left(M\right)\right)$.
This can be seen by looking at Table \ref{tab:Canonical-inner-product-terms}
and noting that each canonical inner product string with $(m-1)\alpha\beta$'s
also has $(m-1)\beta\alpha$'s, so each distinct canonical inner product
string in $SymSqr\left(M\right)$ is associated with a distinct numeric
inner product of $W\times W$. Recall that symbolic squaring uses
the lexicographically lessor to represent locations $(i,j)$ and $(j,i)$
so switching the roles of $\delta_{i}$ and $\delta_{j}$ in column
2 of Table \ref{tab:Canonical-inner-product-terms} don't matter.

Next is to look at how close squaring a widely-space primes matrix
comes to matching a canonical inner product string for off-diagonal
locations after the first symbolic squaring. Start by looking at off-diagonal
locations that are not refined by symbolic squaring.
\begin{lem}
\label{lem:symsqrEqual-WWEqual}Let $W$ be a $m\times m$ widely-spaced
primes matrix. If $SymSqr\left(W\right)_{i,j}=SymSqr\left(W\right)_{r,s}$
then $\left(W\times W\right)_{i,j}=\left(W\times W\right)_{r,s}$.
\end{lem}

\begin{proof}
Given $SymSqr\left(W\right)_{i,j}=SymSqr\left(W\right)_{r,s}$, then
the canonical inner product strings at $(i,j)$ and $(r,s)$ match,
so the multisets of terms match. Therefore, the inner products $\left(W\times W\right)_{i,j}$
and $\left(W\times W\right)_{r,s}$ match since they are summing over
the same multiset of terms.
\end{proof}
Next is to look at off-diagonal locations that do get refined by symbolic
squaring. This is the more complex case. Before addressing the general
case, the instance where column mixes of the inner product vectors
match but the results from symbolic squaring differ is examined.
\begin{thm}
\label{thm:SymSqrNE-WWNE}Let $W$ be a $m\times m$ widely-spaced
primes matrix. Let $(i,j)$ and $(r,s)$ be two distinct off-diagonal
locations. If $colMix\left(\left[W_{i},W_{j}\right]\right)=colMix\left(\left[W_{r},W_{s}\right]\right)$
and $SymSqr\left(W\right)_{i,j}\neq SymSqr\left(W\right)_{r,s}$,
then $\left(W\times W\right)_{i,j}\neq\left(W\times W\right)_{r,s}$.
\end{thm}

\begin{proof}
Assume $W$ is a $m\times m$ widely-spaced primes matrix constructed
as described in Definition \ref{def:widely-spaced-primes} and $(i,j)$
and $(r,s)$ are two distinct off-diagonal locations such that $colMix\left(\left[W_{i},W_{j}\right]\right)=colMix\left(\left[W_{r},W_{s}\right]\right)$
and $SymSqr\left(W\right)_{i,j}\neq SymSqr\left(W\right)_{r,s}$.
WLOG assume that $colMix\left(W_{i}\right)=colMix\left(W_{r}\right)$
and $colMix\left(W_{j}\right)=colMix\left(W_{s}\right)$. 

Proof by contradiction. Assume that $\left(W\times W\right)_{i,j}=\left(W\times W\right)_{r,s}$.
Let $A=P_{ij}WP_{ij}^{T}$ where $P_{ij}$ symmetrically permutes
columns $i$ and $j$ to be columns 1 and 2 of $A$ respectively.
Similarly, let $B=P_{rs}WP_{rs}^{T}$ where $P_{rs}$symmetrically
permutes columns $r$ and $s$ to be columns 1 and 2 of $B$ respectively.

Note that $A$ and $B$ are symmetric and $colMix\left(A_{1}\right)=colMix\left(B_{1}\right)$
and $colMix\left(A_{2}\right)=colMix\left(B_{2}\right)$. So there
exists a permutation $P_{A}$ that when symmetrically applied to $A$,
sorts rows $3:m$ of $A_{1}$ into contiguous groups of distinct symbols.
Similarly there exists a permutation $P_{B}$ that sorts rows $3:m$
of $B_{1}$ into the same contiguous groups of distinct symbols. Let
\[
A\leftarrow P_{A}AP_{A}^{T}\mbox{ and }B\leftarrow P_{B}BP_{B}^{T}.
\]
Now $A_{1}=B_{1}$ and $A$ and $B$ have the form shown in (\ref{eq:differences})
where the row signatures are the groups of distinct symbols in rows
$3:m$ of $A_{1}$ and $B_{1}$. Note that $colMix\left(A_{3:m,2}\right)=colMix\left(B_{3:m,2}\right)$
because $colMix\left(A_{2}\right)=colMix\left(B_{2}\right)$ and $A$
and $B$ are symmetric, so the first two elements of $A_{2}$ and
$B_{2}$ are identical.

By hypothesis $SymSqr\left(W\right)_{i,j}\neq SymSqr\left(W\right)_{r,s}$
so $SymSqr\left(A\right)_{1,2}\neq SymSqr\left(B\right)_{1,2}$ but
$A_{1}^{T}A_{2}=B_{1}^{T}B_{2}$. Consider the terms in the canonical
inner product strings of $SymSqr\left(A\right)_{1,2}$ and $SymSqr\left(B\right)_{1,2}$.
The first two terms involving diagonal symbols match, so the differences
must be in the terms from rows $3:m$, 
\begin{equation}
\left\{ \left\{ \begin{array}{c}
\left(A_{3,1},A_{3,2}\right)\\
\vdots\\
\left(A_{m,1},A_{m,2}\right)
\end{array}\right\} \right\} \neq\left\{ \left\{ \begin{array}{c}
\left(B_{3,1},B_{3,2}\right)\\
\vdots\\
\left(B_{m,1},B_{m,2}\right)
\end{array}\right\} \right\} \label{eq:Multiset3tom}
\end{equation}
where the symmetry of $A$ and $B$ is used to write the left factors
as entries from column 1. Since $A_{1}=B_{1},$ rewrite (\ref{eq:Multiset3tom})
as 
\begin{equation}
\left\{ \left\{ \begin{array}{c}
\left(B_{3,1},A_{3,2}\right)\\
\vdots\\
\left(B_{m,1},A_{m,2}\right)
\end{array}\right\} \right\} \neq\left\{ \left\{ \begin{array}{c}
\left(B_{3,1},B_{3,2}\right)\\
\vdots\\
\left(B_{m,1},B_{m,2}\right)
\end{array}\right\} \right\} .\label{eq:BA-BB}
\end{equation}
 Now the left hand factors on both sides of (\ref{eq:BA-BB}) match.
Also the column mixes match, $colMix\left(A_{3:m,2}\right)=colMix\left(B_{3:m,2}\right)$,
and the inner products match but the multisets of terms differ. Comparing
the lhs and rhs of (\ref{eq:BA-BB}), one sees that $A_{3:m,2}$,
when compared to $B_{3:m,2}$, must have distinct symbols exchanged
across the row signature boundaries established in $B_{1}$. The symbol
exchange(s) act similar to the roles of $\alpha$ and $\beta$ in
(\ref{eq:differences}). This is the only way for the multisets of
terms to differ while $colMix\left(A_{2}\right)$ equals $colMix\left(B_{2}\right)$.
By Lemma \ref{lem:WinnerProdUnique}, the inner product is unique
to the set of distinct terms and their multiplicities. The exchange
of distinct symbols across the row signature boundaries change the
multiplicities of distinct terms and possibly the set of distinct
terms. So, the inner products associated with the two sides of (\ref{eq:BA-BB})
cannot be equal, contradicting the hypothesis that $\left(W\times W\right)_{i,j}=\left(W\times W\right)_{r,s}$.

Therefore, if $colMix\left(\left[W_{i},W_{j}\right]\right)=colMix\left(\left[W_{r},W_{s}\right]\right)$
and $SymSqr\left(W\right)_{i,j}\neq SymSqr\left(W\right)_{r,s}$,
then $\left(W\times W\right)_{i,j}\neq\left(W\times W\right)_{r,s}$
as was to be shown.
\end{proof}
Widely-spaced primes matrices and regular matrix multiplication can
generate a sequence of refinement that is at least as refined as the
one generating $\Pi\left(M^{*}\right)$, it just might take twice
as many iterations.
\begin{thm}
\label{thm:TwiceIsEnough}For a widely-spaced primes matrix $W$,
let $\overline{W}=wspm\left(W\times W\right)$ be a widely-spaced
primes matrix generated from the result of squaring $W$. Then the
pattern from squaring $\overline{W}$ is a refinement of the pattern
from symbolically squaring $W$,

\[
\Pi\left(\overline{W}\times\overline{W}\right)\preceq\Pi\left(SymSqr\left(W\right)\right).
\]
\end{thm}

\begin{proof}
Let $W$ be a widely-spaced primes matrix and $\overline{W}=wspm\left(W\times W\right)$
be a widely-spaced primes matrix generated from the result of squaring
$W$.

From Theorem \ref{thm:WrefinesDiagaonl} and Theorem \ref{thm:SPDMonotonic-1}
one gets that, 
\[
\Pi\left(diag\left(\overline{W}\times\overline{W}\right)\right)\preceq\Pi\left(diag\left(W\times W\right)\right)=\Pi\left(diag\left(SymSqr\left(W\right)\right)\right).
\]
 So diagonal cells of $\left(\overline{W}\times\overline{W}\right)$
are a refinement of the diagonal cells of $SymSqr\left(W\right)$.
The remaining question is whether the off-diagonal cells of $\left(\overline{W}\times\overline{W}\right)$
are a refinement of the off-diagonal cells of $SymSqr\left(W\right)$.

Assume $W_{i,j}$ and $W_{r,s}$ are in the same off-diagonal cell
so $W_{i,j}=W_{r,s}$ with $(i,j)\neq(r,s)$. 
\begin{casenv}
\item If $SymSqr\left(W\right)_{i,j}\neq SymSqr\left(W\right)_{r,s}$ and
$\left(W\times W\right)_{i,j}\neq\left(W\times W\right)_{r,s}$, then
done since $\left(\overline{W}\times\overline{W}\right)_{i,j}\neq\left(\overline{W}\times\overline{W}\right)_{r,s}$
by Theorem \ref{thm:SPDMonotonic-1}.
\item If $SymSqr\left(W\right)_{i,j}\neq SymSqr\left(W\right)_{r,s}$ and
$\left(W\times W\right)_{i,j}=\left(W\times W\right)_{r,s}$ then
consider the column mixes of $W_{i}$, $W_{j}$, $W_{r}$, and $W_{s}$.
\begin{casenv}
\item Assume $colMix\left(\left[W_{i},W_{j}\right]\right)\neq colMix\left(\left[W_{r},W_{s}\right]\right)$.
Theorem \ref{thm:WrefinesDiagaonl} says $W\times W$ and $SymSqr\left(W\right)$
refine the diagonal the same way, so the multisets of diagonal symbols
are not equal $\left\{ \left\{ \left(W\times W\right)_{i,i},\left(W\times W\right)_{j,j}\right\} \right\} \neq\left\{ \left\{ \left(W\times W\right)_{r,r},\left(W\times W\right)_{s,s}\right\} \right\} $.
Then the second squaring results in $\left(\overline{W}\times\overline{W}\right)_{i,j}\neq\left(\overline{W}\times\overline{W}\right)_{r,s}$
since the diagonal terms $\left\{ \left\{ \overline{W}_{i,i},\overline{W}_{j,j}\right\} \right\} $
and $\left\{ \left\{ \overline{W}_{r,r},\overline{W}_{s,s}\right\} \right\} $
differ in (\ref{eq:DOODij}) and (\ref{eq:DOODrs}).
\item Assume $colMix\left(\left[W_{i},W_{j}\right]\right)=colMix\left(\left[W_{r},W_{s}\right]\right)$.
Theorem \ref{thm:SymSqrNE-WWNE}, says the condition where $colMix\left(\left[W_{i},W_{j}\right]\right)=colMix\left(\left[W_{r},W_{s}\right]\right)$
and $SymSqr\left(W\right)_{i,j}\neq SymSqr\left(W\right)_{r,s}$ means
$\left(W\times W\right)_{i,j}\neq\left(W\times W\right)_{r,s}$. This
contradicts the hypothesis that $\left(W\times W\right)_{i,j}=\left(W\times W\right)_{r,s}$.
So this is not a valid case.
\end{casenv}
\end{casenv}
Therefore, $\Pi\left(\overline{W}\times\overline{W}\right)\preceq SymSqr\left(W\right)$
as was to be shown.
\end{proof}
Theorem \ref{thm:TwiceIsEnough} implies a sequence of constructing
and squaring widely-spaced matrices will converge to a fixed point
that is a refinement of $M^{*}$ in at most twice as many iterations
as symbolic squaring. 

A second method generates a sequence of refinement identical to the
symbolic squaring sequence of refinement. However, instead of squaring
a widely-spaced primes matrix, it multiplies two different widely-space
matrices. 
\begin{thm}
\label{thm:N1xN2=00003DsymSqr(N)}Given a real $m\times m$ symmetric
matrix $M$ whose diagonal symbols are distinct from off-diagonal
symbols, there exist widely-spaced primes matrices $W1$ and $W2$
such that $\Pi\left(W1\right)=\Pi\left(W2\right)=\Pi\left(M\right)$
and 
\[
\Pi\left(min\left(\left(W1\times W2\right),\left(W1\times W2\right)^{T}\right)\right)=\Pi\left(SymSqr\left(M\right)\right)
\]
 where $min\left(\right)$ is the element-wise matrix minimum.
\end{thm}

\begin{proof}
Similar to the proof of Theorem \ref{thm:SPDMonotonic-1}, widely-spaced
primes are used as symbols. Assume $M$ has $n$ distinct symbols
and $n_{1}$ are off-diagonal symbols and $n_{2}$ are diagonal symbols
with $n_{1}+n_{2}=n$. Let $W1=wspm\left(M\right)$ be a widely-space
primes matrix as described in Definition \ref{def:widely-spaced-primes}
where the primes are $\left\{ p_{1},\cdots,p_{n}\right\} $, starting
with $p_{1}>m$. Then let $W2=wspm\left(M\right)$ be a second widely-spaced
primes matrix where the primes are $\left\{ \overline{p}_{1},\cdots,\overline{p}_{n}\right\} $,
starting with $\overline{p}_{1}>mp_{n}^{2}$.

Use the first $n_{1}$ primes of $\left\{ p_{1},\cdots,p_{n}\right\} $,
$p_{1},\ldots,p_{n_{1}}$, as off-diagonal symbols of $W1$ and the
next $n_{2}$ primes, $p_{n_{1}+1},\ldots,p_{n}$, as the diagonal
symbols of $W1$. Then use the first $n_{1}$ primes of $\left\{ \overline{p}_{1},\cdots,\overline{p}_{n}\right\} $,
$\overline{p}_{1},\ldots,\overline{p}_{n_{1}}$, as off-diagonal symbols
of $W2$ and the last $n_{2}$ primes, $\overline{p}_{n_{1}+1},\ldots,\overline{p}_{n}$,
as diagonal symbols of $W2$. $W1$ and $W2$ are SPD by construction.

Now consider the matrix product $W1\times W2$. For each $(i,j)$
location, the inner product is constructed from a row vector of $W1$
and a column vector from $W2$. Since the symbol sets for $W1$ and
$W2$ are distinct, the factors of each term can be considered to
be ordered. To see this, consider the canonical string at $(i,j)$
from $W1\times W2$. Each term in an inner product is given by a $p_{r}\overline{p}_{c}$
where $p_{r}\in\left\{ p_{1},\ldots,p_{n}\right\} $ and $\overline{p}_{c}\in\left\{ \overline{p}_{1},\cdots,\overline{p}_{n}\right\} $.
Without loss of generality, assume that $p_{r}\overline{p}_{c}$ is
the largest term in the inner product inner product string. Then the
inner product is less than or equal to $mp_{r}\overline{p}_{c}$.
But the next larger distinct factor after $p_{r}\overline{p}_{c}$
is greater than $mp_{r}^{2}\overline{p}_{c}$. So the inner product
can be decomposed using modulo arithmetic starting with the largest
prime and working downward to determine the exact number of times
each term appears. This implies distinct canonical inner product strings
in $SymSqr\left(M\right)$ have distinct inner products in $W1\times W2$.

The remaining part is to note that $SymSqr\left(M\right)$ selects
the lessor of the canonical strings at $(i,j)$ and $(j,i)$ to represent
both locations for symmetric matrices. For $W1\times W2$, this is
accomplished by using the minimum of $\left[W1\times W2\right]_{i,j}$
and $\left[W1\times W2\right]_{j,i}$ to represent both locations.
Recall that the product of two real symmetric matrices is not symmetric
unless they have the same eigenvectors. So 
\[
\Pi\left(min\left(\left(W1\times W2\right),\left(W1\times W2\right)^{T}\right)\right)=\Pi\left(SymSqr\left(M\right)\right)
\]
 as desired.
\end{proof}
So the refinement sequence generated by symbolic squaring can be reproduced
using regular matrix multiplication. However, it involves multiplying
widely-spaced primes matrices. 

One consequence of replacing symbolic squaring with regular matrix
multiplication on widely-spaced primes matrices is that the bounds
for Algorithm \ref{alg:Blind-p-similarity-algorithm} reduces from
$O\left(m^{8}\right)$ for symbolic squaring to $O\left(m^{6}\right)$
for matrix multiplication of widely-spaced primes matrices and from
$O\left(m^{10}\right)$ for symbol substitution of canonical inner
product string arrays to $O\left(m^{8}\right)$ for symbol substitution
of positive integer matrices. Dropping the overall complexity of the
algorithm from $O\left(m^{12}\right)$ to $O\left(m^{10}\right)$. 

\subsection{Symbolic ``Walks'' }

Two vertices that are not neighbors in a rook's graph are part of
a unique four-cycle \cite{moon63}. This is also true of PCGs since
the connectivity, ignoring edge weights, is identical to a rook's
graph. From Theorem \ref{thm:FirstDMD} the first symbolic squaring
has a pattern identical to the pattern from symbolically multiplying
$DMD$ where the $m^{2}\times m^{2}$ matrix $M=R+D$ is a PCM. Consider
the value of $DMD$ at $(i,j)$. It is given by 

\[
\left[DMD\right]_{i,j}=\begin{cases}
D_{ii}R_{ij}D_{jj} & \mbox{if \ensuremath{i\neq j}}\\
D_{ii}^{3} & \mbox{if \ensuremath{i=j}}
\end{cases}
\]
which, for $i\neq j$, may be interpreted as the relationship between
the pair of vertices $D_{ii}$ and $D_{jj}$.

Now symbolically square $DMD$ for the second symbolic squaring. We
have 
\begin{eqnarray*}
\left(DMD\right)\left(DMD\right) & = & \left(D\left(R+D\right)D\right)\times\left(D\left(R+D\right)D\right)\\
 & = & DRD^{2}RD+DRD^{4}+D^{4}RD+D^{6}.
\end{eqnarray*}
Focus on the $DRD^{2}RD$ term, symbolically, 
\[
\Pi\left(SymMult\left(D,R,D^{2},R,D\right)\right)\preceq\Pi\left(DRD^{4}+D^{4}RD+D^{6}\right)
\]
 if the zeros on the diagonal of $R$ are treated as symbols. Implying
that 
\[
\Pi\left(SumMult\left(D,R,D^{2},R,D\right)\right)=\Pi\left(SymSqr\left(DMD\right)\right).
\]
 One sees that its value at $(i,j)$, $i\neq j$, is equal to 
\[
\left[DRD^{2}RD\right]_{ij}=\sum_{k=1}^{m^{2}}D_{ii}R_{ik}D_{kk}^{2}R_{kj}D_{jj}
\]
 where each term is a ``walk'' of length two, with segments from
vertex $D_{ii}$ to vertex $D_{kk}$ and then from $D_{kk}$ to vertex
$D_{jj}$. Each ``walk'' is a term in the canonical inner product.
``Walk'' is in quotes since a walk normally implies an edge connecting
the vertices. Here ``walk'' is being used more generally to represent
a string of relationships between vertices.

For the third symbolic squaring, $\left(DMD\right)^{4}$ has a term
$DRD^{2}RD^{2}RD^{2}RD$ whose symbolic inner product terms for $(i,j)$,
$i\neq j$, are given by 
\begin{equation}
\left[DRD^{2}RD^{2}RD^{2}RD\right]_{ij}=\sum_{k_{1}=1}^{m^{2}}\sum_{k_{2}=1}^{m^{2}}\sum_{k_{3}=1}^{m^{2}}D_{ii}R_{ik_{1}}D_{k_{1}k_{1}}^{2}R_{k_{1}k_{2}}D_{k_{2}k_{2}}^{2}R_{k_{2}k_{3}}D_{k_{3}k_{3}}^{2}R_{k_{3}j}D_{jj}.\label{eq:4Cycles}
\end{equation}
Now let $i$ equal to $j$, (\ref{eq:4Cycles}) contains all of the
four-cycles as well as other ``walks'' of length four. For the $n^{{\rm th}}$
symbolic squaring, there are $m^{2n}$ ``walks'' of length $2^{(n-1)}$
at each $(i,j)$ location. 

Section \ref{subsec:Origin-of-PCGs} described an unsuccessful attempt
to just use the four-cycles to separate graphs. However, it is also
pointed out that there must be unique row and column signatures that
can be used to pin down the value at a location. As symbolic squaring
continues, the ``walks'' will incorporate the row and column signatures
as well as the value at the intersection.

\subsection{Practical Considerations}

From a practical point of view, Algorithm \ref{alg:Blind-p-similarity-algorithm}
is extremely slow and requires a large amount of memory. A $m\times m$
matrix has a $m^{2}\times m^{2}$ PCM which is symbolically dense.
Testing used adjacency matrices of graphs from the Bliss graph collection
\cite{bliss} and the Neuen-Schweitzer graph database \cite{neuen2017b}.
Only matrices up to $180\times180$ were tested since their PCMs are
$32400\times32400$ and take up about 8GB of memory per PCM when using
64-bit values. To make the execution time tractable, double precision
matrix multiplication with prime numbers as the symbols is used instead
of symbolic matrix multiplication. Symbol substitutions are performed
every iteration. When there are more than 10K symbols the symbol substitution
algorithm switches to using the integers $1,\ldots,k$ as the symbols
to reduce magnitudes. IEEE 754 double precision format has a 52 bit
mantissa which can represent integers up to 15 digits long. So the
process is halted if an inner product goes over $2^{52}$. Lastly
the heuristic compares column mixes instead of diagonal mixes since
column mixes often detect non p-similar cases an iteration before
the diagonal mixes as discussed in Section \ref{subsec:Comparing-Multisets}.

\section{\label{sec:Summary}Summary}

Permutation similarity between two (real or complex) square matrices
reduces to asking whether their associated PCGs are isomorphic. PCGs
add symmetric permutation constraints to matrices so multisets can
be used for comparisons. Usually, multisets maintain either row or
column constraints while giving up the other constraint. PCGs make
row and column constraints an integral part of the problem. Then the
associated PCMs are repeatedly symbolically squared and symbols substituted
until stable partitions are reached or the diagonal mixes differ.
If the diagonal mixes do not match, the matrices are not p-similar.
The case where stable partitions are reached with matching diagonal
mixes are p-similar. The necessity and sufficiency of the algorithm
is shown. A second algorithm using the blind algorithm to find the
permutation between two p-similar matrices is also given. Therefore
permutation similarity and graph isomorphism are in P. 

\bibliographystyle{plain}
\bibliography{References/giReferences}

\end{document}